\documentclass[a4paper,12pt]{article}

\usepackage{amsmath,amsthm,amssymb,multicol,epsf}
\usepackage{color}
\usepackage{graphicx}
\usepackage{amsmath}
\usepackage{amsfonts}
\usepackage{amssymb}
\usepackage{amsthm}
\usepackage{epstopdf}
\usepackage{algorithmic}
\usepackage{setspace}
\usepackage[cm]{fullpage}
\usepackage{multirow}

\newtheorem{theorem}{Theorem}[section]
\newtheorem{lem}{Lemma}[section]

\newtheorem{remark}{Remark}[section]

\theoremstyle{remark}
\newtheorem{algorithm}{Algorithm}

\newcommand{\mbf}[1]{\mbox{\boldmath$#1$}}

\numberwithin{equation}{section}

\title{Subdifferential-based implicit return-mapping operators in Mohr-Coulomb plasticity}
\author{S. Sysala$^1$, M. Cermak$^{2}$ \\ \\ $^1$Institute of Geonics, Czech Academy of Sciences, Ostrava, Czech
Republic\\ $^2$V\v SB--Technical University of Ostrava, Ostrava, Czech Republic}

\begin{document}

\maketitle

\begin{abstract}
The paper is devoted to a constitutive solution, limit load analysis and Newton-like methods in elastoplastic problems containing the Mohr-Coulomb yield criterion. Within the constitutive problem, we introduce a self-contained derivation of the implicit return-mapping solution scheme using a recent subdifferential-based treatment. Unlike conventional techniques based on Koiter's rules, the presented scheme a priori detects a position of the unknown stress tensor on the yield surface even if the constitutive solution cannot be found in closed form. This fact eliminates blind guesswork from the scheme, enables to analyze properties of the constitutive operator, and simplifies construction of the consistent tangent operator which is important for the semismooth Newton method applied on the incremental boundary value elastoplastic problem. The incremental problem in Mohr-Coulomb plasticity is combined with the limit load analysis. Beside a conventional direct method of the incremental limit analysis, a recent indirect one is introduced and its advantages are described. The paper contains 2D and 3D numerical experiments on slope stability with publicly available Matlab implementations.

\end{abstract}


\section{Introduction}

This paper is a continuation of \cite{SCKKZB15} which was devoted to a solution of elastoplastic constitutive problems using a subdifferential formulation of the plastic flow rule. It leads to simpler and more correct implicit constitutive solution schemes. While a broad class of elastoplastic models containing 1 or 2 singular points (apices) on the yield surface was considered in \cite{SCKKZB15}, the aim of this paper is to approach the subdifferential-based treatment to models that are usually formulated in terms of principal stresses.  For example, the principal stresses are used in models containing the Mohr-Coulomb, the Tresca, the Rankine, the Hoek-Brown or the unified strength yield criteria \cite{NPO08,CDA15,LL15,LR96}. Such criteria have a {\it multisurface} representation leading to a relatively complex structure of singular points. 

Due to technical complexity of implicit solution schemes for these models, we focus only on a particular but representative yield criterion: the Mohr-Coulomb one. This criterion is broadly exploited in soil and rock mechanics and its surface is a hexagonal pyramid aligned with the hydrostatic axis (see, e.g., \cite{NPO08}). We consider the Mohr-Coulomb model introduced in \cite[Section 8]{NPO08} which can optionally contain the {\it nonassociative flow rule} and the {\it nonlinear isotropic hardening}. The nonassociative flow rule enables to catch the dilatant behavior of a material. Further, due to the presence of the nonlinear hardening, one cannot find the implicit constitutive solution in closed form, and thus the problem remains challenging. As in \cite{NPO08}, we let a hardening function in an abstract form. For a particular example of the nonlinear hardening in soil mechanics, we refer, e.g., \cite{BSS03}.

In literature, there are many various concepts of the constitutive solution schemes for models containing yield criteria written in terms of the principal stresses. For their detailed overview and historical development, we refer the recent papers \cite{CDA15} and \cite{K13}, respectively. It is worth mentioning that the solution schemes mainly depend on a formulation of the plastic flow rule, its discretization and other eventual approximations.

In engineering practice, the plastic flow rule is usually formulated using the so-called Koiter rule introduced in \cite{K53} for associative models with multisurface yield criteria. Consequently, this rule was also extended for nonassociative models, see, e.g., \cite{dB87}. It consists of several formulas that depend on a position of the unknown stress tensor $\mbf\sigma$ on the yield surface. The formulas have a different number of plastic multipliers. Within the Mohr-Coulomb pyramid, one plastic multiplier is used for smooth portions, two multipliers at edge points, and six multipliers at the apex. For each Koiter's formula, a different solution scheme is introduced. However,  only one of which usually gives the correct stress tensor. Moreover, the handling with different numbers of plastic multipliers is not suitable for analysing the stress-strain operator even if the solution can be found in closed form. If an elastoplastic model contains a {\it convex plastic potential} as the Mohr-Coulomb one then it is possible to replace the Koiter rule with a subdifferential of the potential (see, e.g., \cite{NPO08}). Such a formulation is independent of the unknown stress position, contains just one plastic multiplier, and thus it is more convenient for mathematical analysis of the constitutive operators. In \cite{SCKKZB15}, it was shown that this formulation is also convenient for a solution of some constitutive problems. Further, in some special cases, the constitutive problem can  be also defined using the principle of maximum plastic dissipation \cite{NPO08, HR99} or by the theory of bipotentials \cite{B12} and solved by techniques based on mathematical programming. 

We focus on the (fully) implicit Euler discretization of the flow rule, which is frequently used in elastoplasticity. Beside other Euler-type methods (see, e.g., \cite{NPO08, SH98}), the cutting plane methods are also popular. We refer, e.g., \cite{SHVS14} for the literature survey and recent development of these methods. When the constitutive problems are discretized by the implicit Euler methods, the solution is searched by the elastic predictor -- plastic correction method. Within the plastic correction, the so-called (implicit) return-mapping scheme is constructed. It is worth mentioning that plastic correction problems can be reduced to problems formulated only in terms of the principal stresses \cite{CDA06, CDA15, NPO08}.

In order to simplify the solution schemes for nonsmooth yield criteria, many various approximative techniques have been suggested. These techniques are based on local or global smoothing of yield surfaces or plastic potentials. For literature survey, we refer \cite[Section 1.2]{CDA15} or \cite{ALSH11, B13, BSS03}. However, such an approach is out of the scope of this paper.

The constitutive problem is an essential part of the overall initial boundary value elastoplastic problem. Its time discretization leads to the incremental boundary value problem which is mostly solved by nonsmooth variants of the Newton method \cite{CKSV14, GrVa09, SaWi11, Sy09, Sy14} in each time step. Then, it is useful to construct the so-called consistent tangent operator representing a generalized derivative of the discretized constitutive stress-strain operator. We use the framework based on the eigenprojections of symmetric second order tensors, see, e.g., \cite{CaHo86, NPO08}. A similar approach is also used in the recent book \cite{B13} with slightly different terminology like the spectral directions or the spin of a tensor. Another approach is introduced, e.g., in \cite{C97, CDA06, CDA15} where the consistent tangent operator is determined by the tangent operator representing the relation between the stress and strain rates.

Further, this paper is devoted to the limit load problem which is frequently combined with the Mohr-Coulomb model. It is an additional problem to the elastoplastic one where the load history is not fully prescribed. It is only given a fixed external force that is multiplied by a scalar load parameter whose limit value is unknown. It is well known that the investigated body collapses when this critical value is exceeded. Therefore, this value is an important safety parameter and beyond it no solution exists. Strip-footing collapse or slope stability are traditional applications on this problematic (see, e.g. \cite{CL90, NPO08}).  The simplest computational technique is based on the so-called {\it incremental limit analysis} where the load parameter is enlarged up to its limit value. Then, the boundary-value elastoplastic problem is solved for investigated values of this parameter. Beside the conventional {\it direct method} of the incremental limit analysis, we also introduce the {\it indirect method} and describe its advantages based on recent expertise introduced in \cite{SHHC15, CHKS15, HRS16, HRS16b}.

The rest of the paper is organized as follows. In Section \ref{sec_spectrum}, an auxilliary framework related to the subdifferential of an eigenvalue function and derivatives of eigenprojections is introduced. In Section \ref{sec_model}, the Mohr-Coulomb constitutive initial value problem is formulated using the subdifferential of the plastic potential and discretized by the implicit Euler method. In Section \ref{sec_time_discret}, the existence and uniqueness of a solution to the discretized problem is proven and the improved solution scheme is derived. In Section \ref{sec.Stress-strain_relation}, the stress-strain and the consistent tangent operators are constructed. In Section \ref{sec_realization}, the direct and indirect methods of the incremental limit analysis are introduced. Both methods are combined with the semismooth Newton method. In Section \ref{sec_experiments}, 2D and 3D numerical experiments related to slope stability are introduced. In Section \ref{sec_conclusion}, some concluding remarks are mentioned. The paper also contains Appendix with some useful auxilliary results. In Appendix A, the solution scheme is simplified under the plane strain assumptions. In Appendix B, algebraic representations for second and fourth order tensors within the 3D and plane strain problems are derived.

In this paper, second order tensors, matrices, and vectors are denoted by
bold letters.  Further, the fourth order tensors
are denoted by capital blackboard letters, e.g., $\mathbb D_e$ or $\mathbb
I$. The symbol $\otimes$ means the tensor product \cite{NPO08}. We also use the following
notation: $\mathbb R_+:=\{z\in\mathbb R;\; z\geq0\}$ and $\mathbb R^{3\times
3}_{sym}$ for the space of symmetric, second order tensors. The standard scalar product in $\mathbb R^{3}$ and the biscalar product in $\mathbb R^{3\times
3}_{sym}$ are denoted as $\cdot$ and $:$, respectively.

%%%%%%%%%%%          Section 2         %%%%%%%%%%%%%%%%%%%%%%%%%%%%%%%%%%%%%%%%%%%%%%%%%%%
\section{Subdifferentials and derivatives of eigenvalue functions}
\label{sec_spectrum}

In this section, we introduce an auxilliary framework that will be crucial for an efficient construction of the constitutive and consistent tangent operators in Mohr-Coulomb plasticity. Let
\begin{equation}
\mbf\eta=\sum_{i=1}^3 \eta_i\mathbf e_i\otimes\mathbf e_i,\quad \eta_1\geq \eta_2\geq \eta_3,
\label{spectrum}
\end{equation}
be the spectral decomposition of a tensor $\mbf\eta\in\mathbb R^{3\times3}_{sym}$. Here, $\eta_i\in\mathbb R$, $\mathbf e_i\in\mathbb R^3$, $i=1,2,3$, denote the eigenvalues, and the eigenvectors of $\mbf\eta$, respectively. The eigenvalues  $\eta_1,\eta_2,\eta_3$ can be computed using the Haigh-Westargaard coordinates (see, e.g., \cite[Appendix A]{NPO08}), and they are uniquely determined with respect to the prescribed ordering. Let $\omega_1,\omega_2,\omega_3$ denote the corresponding eigenvalue functions, i.e. $\eta_i:=\omega_i(\mbf\eta)$, $i=1,2,3$. Further, we define the following set of admissible eigenvectors of $\mbf\eta$: 
$$V(\mbf\eta)=\{(\mathbf e_1,\mathbf e_2,\mathbf e_3)\in \mathbb R^3\times\mathbb R^3\times\mathbb R^3\ |\; \mathbf e_i\cdot\mathbf e_j=\delta_{ij};\; \mbf\eta\mathbf e_i=\eta_i\mathbf e_i,\; i,j =1,2,3;\; \eta_1\geq \eta_2\geq \eta_3\}.$$

\subsection{Subdifferential of an eigenvalue function}
\label{subsec_subdif}

Recall the definition of the subdifferential to a convex function $g:\mathbb R^{3\times 3}_{sym}\rightarrow \mathbb R$ at $\mbf\eta$:
$$
\partial g(\mbf\eta)=\{\mbf\nu\in\mathbb R^{3\times 3}_{sym}\ |\;  g(\mbf\tau)\geq g(\mbf\eta)+\mbf\nu:(\mbf\tau-\mbf\eta)\;\;\forall \mbf\tau\in\mathbb R^{3\times 3}_{sym}\}.
$$
To receive the Mohr-Coulomb yield function or the plastic potential, we specify $g$ as follows:
\begin{equation}
g(\mbf\eta)=a\omega_1(\mbf\eta)-b\omega_3(\mbf\eta),\quad  \mbf\eta\in\mathbb R^{3\times3}_{sym},
\label{g_ab}
\end{equation}
where the parameters $a,b\geq0$ are sufficiently chosen. Notice that the convexity of the eigenvalue function $g$ can be derived from:
\begin{equation}
\omega_1(\mbf\eta)=\max_{\substack{\mathbf e\in\mathbb R^3\\ |\mathbf e|=1}}\mbf\eta:(\mathbf e\otimes\mathbf e)=\max_{\substack{\mathbf e\in\mathbb R^3\\ |\mathbf e|=1}}\ (\mbf\eta\mathbf e)\cdot\mathbf e,\quad\omega_3(\mbf\eta)=\min_{\substack{\mathbf e\in\mathbb R^3\\ |\mathbf e|=1}}\mbf\eta:(\mathbf e\otimes\mathbf e).
\label{eig_maxmin}
\end{equation}
Specific form of $\partial g(\mbf\eta)$ with respect to (\ref{g_ab}) can be found using a framework introduced in \cite[Chapter 2]{Ru06}. We derive another form of $\partial g(\mbf\eta)$ that is convenient for purposes of this paper.

\begin{lem}
Let $g:\mathbb R^{3\times 3}_{sym}\rightarrow \mathbb R$ be defined by (\ref{g_ab}). Then for any $\mbf\eta\in\mathbb R^{3\times 3}_{sym}$, it holds:
\begin{eqnarray}
\partial g(\mbf\eta)&=&\left\{ \mbf\nu=\sum_{i=1}^3\nu_i\mathbf e_i\otimes\mathbf e_i\in\mathbb R^{3\times 3}_{sym}\ |\; (\mathbf e_1, \mathbf e_2, \mathbf e_3)\in V(\mbf\eta); \; a\geq\nu_1\geq\nu_2\geq\nu_3\geq -b;\right.\nonumber\\
&&\quad \left.\sum_{i=1}^3\nu_i=a-b;\; (\nu_1-a)[\omega_1(\mbf\eta)-\omega_2(\mbf\eta)]=0;\; (\nu_3+b)[\omega_2(\mbf\eta)-\omega_3(\mbf\eta)]=0 \right\}.
\label{sub_g_ab}
\end{eqnarray}
\label{lem_subdif}
\end{lem}

\begin{proof}
Since $g(\mbf 0)=0$ and $g(2\mbf\eta)=2g(\mbf\eta)$ the standard definition of $\partial g(\mbf\eta)$ is equivalent to:
\begin{equation}
\partial g(\mbf\eta)=\{\mbf\nu\in\mathbb R^{3\times 3}_{sym}\ |\; g(\mbf\eta)=\mbf\nu:\mbf\eta;\; g(\mbf\tau)\geq\mbf\nu:\mbf\tau\;\;\forall \mbf\tau\in\mathbb R^{3\times 3}_{sym}\}.
\label{sub_def}
\end{equation}
First, we derive necessary and sufficient conditions on $\mbf\nu\in\mathbb R^{3\times 3}_{sym}$ ensuring
\begin{equation}
g(\mbf\tau)\geq\mbf\nu:\mbf\tau\quad \forall \mbf\tau\in\mathbb R^{3\times 3}_{sym}.
\label{ineq_g_T}
\end{equation}
To this end, consider the following spectral decomposition of $\mbf\nu$:
\begin{equation}
\mbf\nu=\sum_{i=1}^3\nu_i\mathbf f_i\otimes\mathbf f_i,\quad \nu_1\geq\nu_2\geq\nu_3,\quad (\mathbf f_1,\mathbf f_2,\mathbf f_3)\in V(\mbf\nu).
\label{cond_sub_0}
\end{equation}
Choose $\mbf\tau=\pm\mbf I$, where $\mbf I$ is the unit tensor in $\mathbb R^{3\times 3}_{sym}$. Then from (\ref{ineq_g_T}), (\ref{cond_sub_0}) we have:
\begin{equation}
\nu_1+\nu_2+\nu_3=a-b.
\label{cond_sub_1}
\end{equation}
Choose $\mbf\tau=\mathbf f_1\otimes\mathbf f_1$ and $\mbf\tau=-\mathbf f_3\otimes\mathbf f_3$.  Then from (\ref{ineq_g_T}), (\ref{cond_sub_0}) we derive, respectively:
\begin{equation}
\nu_1\leq a,\quad \nu_3\geq -b.
\label{cond_sub_2}
\end{equation}
Let $\mbf\tau\in\mathbb R^{3\times 3}_{sym}$ be arbitrarily chosen and denote $\tau_i:=\mbf\tau:(\mathbf f_i\otimes\mathbf f_i)$, $i=1,2,3$, $(\mathbf f_1,\mathbf f_2,\mathbf f_3)\in V(\mbf\nu)$. Then,
\begin{equation}
\tau_1+\tau_2+\tau_3=\mbf\tau:\mbf I=\omega_1(\mbf\tau)+\omega_2(\mbf\tau)+\omega_3(\mbf\tau),\qquad \omega_1(\mbf\tau)\geq\tau_i\geq\omega_3(\mbf\tau),\quad\forall i=1,2,3,
\label{tau_i_prop}
\end{equation}
follow from $\mbf I=\sum_{i=1}^3\mathbf f_i\otimes\mathbf f_i$ and (\ref{eig_maxmin}), respectively. Consequently,
\begin{eqnarray}
\mbf\nu:\mbf\tau&=&\sum_{i=1}^3\nu_i\tau_i=\tau_1(\nu_1- \nu_2)+(\tau_1+\tau_2)(\nu_2-\nu_3)+(\tau_1+\tau_2+\tau_3)\nu_3\nonumber\\
&\stackrel{(\ref{tau_i_prop})}{=}&\tau_1(\nu_1- \nu_2)+(\mbf\tau:\mbf I-\tau_3)(\nu_2-\nu_3)+\nu_3\mbf\tau:\mbf I\nonumber\\
&\stackrel{(\ref{tau_i_prop})}{\leq}&\omega_1(\mbf\eta)(\nu_1- \nu_2)+[\mbf\tau:\mbf I-\omega_3(\mbf\eta)](\nu_2-\nu_3)+\nu_3\mbf\tau:\mbf I\nonumber=\sum_{i=1}^3\nu_i\omega_i(\mbf\eta)\\
&=&\nu_1[\omega_1(\mbf\eta)-\omega_2(\mbf\eta)]+(\nu_1+\nu_2)[\omega_2(\mbf\eta)-\omega_3(\mbf\eta)]+(\nu_1+\nu_2+\nu_3)\omega_3(\mbf\eta)\nonumber\\
&\stackrel{(\ref{cond_sub_1})}{=}&\nu_1[\omega_1(\mbf\eta)-\omega_2(\mbf\eta)]+(a-b-\nu_3)[\omega_2(\mbf\eta)-\omega_3(\mbf\eta)]+(a-b)\omega_3(\mbf\eta)\nonumber\\
&\stackrel{(\ref{cond_sub_2})}{\leq}&a[\omega_1(\mbf\eta)-\omega_2(\mbf\eta)]+a[\omega_2(\mbf\eta)-\omega_3(\mbf\eta)]+(a-b)\omega_3(\mbf\eta)\nonumber\\
&=&a\omega_1(\mbf\tau)-b\omega_3(\mbf\tau)=g(\mbf\tau)\quad\forall \mbf\tau\in\mathbb R^{3\times 3}_{sym}.
\label{ineq_T}
\end{eqnarray}
Thus the conditions (\ref{cond_sub_0})-(\ref{cond_sub_2}) are necessary and sufficient for (\ref{ineq_g_T}).

Secondly, assume that $\mbf\nu$ belongs to $\partial g(\mbf\eta)$. Then (\ref{cond_sub_0})-(\ref{cond_sub_2}) hold. Since  $g(\mbf\eta)\stackrel{(\ref{sub_def})}{=}\mbf\nu:\mbf\eta$, the equalities must hold within the derivation of (\ref{ineq_T}) for $\mbf\tau=\mbf\eta$, i.e., we have:
\begin{equation}
(\tau_1-\omega_1(\mbf\eta))(\nu_1- \nu_2)=0,\quad (\tau_3-\omega_3(\mbf\eta))(\nu_2- \nu_3)=0,
\label{cond_sub_4a}
\end{equation}
\begin{equation}
(\nu_1-a)[\omega_1(\mbf\eta)-\omega_2(\mbf\eta)]=0,\quad (\nu_3+b)[\omega_2(\mbf\eta)-\omega_3(\mbf\eta)]=0.
\label{cond_sub_7}
\end{equation}
It is easy to see that the equalities in (\ref{cond_sub_4a}) imply:
\begin{equation}
\exists (\mathbf e_1, \mathbf e_2, \mathbf e_3)\in V(\mbf\eta):\quad \mbf\nu=\sum_{i=1}^3\nu_i\mathbf e_i\otimes\mathbf e_i.
\label{cond_sub_8}
\end{equation}

We have proven that for any element $\mbf\nu\in\partial g(\mbf\eta)$ the conditions (\ref{cond_sub_0})-(\ref{cond_sub_2}), (\ref{cond_sub_7}) and (\ref{cond_sub_8}) hold. Therefore,
\begin{eqnarray}
\partial g(\mbf\eta)&\subset&\left\{ \mbf\nu=\sum_{i=1}^3\nu_i\mathbf e_i\otimes\mathbf e_i\in\mathbb R^{3\times 3}_{sym}\ |\; (\mathbf e_1, \mathbf e_2, \mathbf e_3)\in V(\mbf\eta); \; a\geq\nu_1\geq\nu_2\geq\nu_3\geq -b;\right.\nonumber\\
&&\quad \left.\sum_{i=1}^3\nu_i=a-b;\; (\nu_1-a)[\omega_1(\mbf\eta)-\omega_2(\mbf\eta)]=0;\; (\nu_3+b)[\omega_2(\mbf\eta)-\omega_3(\mbf\eta)]=0 \right\}.\qquad
\label{sub_g_ab2}
\end{eqnarray}
Conversely, one can easily check that any element from the set on the right hand side in (\ref{sub_g_ab2}) belongs to $\partial g(\mbf\eta)$ using (\ref{sub_def}) and (\ref{ineq_T}).
\end{proof}

\begin{remark}
\emph{One can easily specify the eigenvalues $\nu_1$, $\nu_2$ and $\nu_3$ in (\ref{sub_g_ab}) depending on a number of distinct eigenvalues of  $\mbf\eta$. If $\eta_1>\eta_2>\eta_3$ then $\nu_1=a$, $\nu_2=0$ and $\nu_3=-b$. If $\eta_1=\eta_2>\eta_3$ then $a\geq\nu_1\geq\nu_2\geq0$, $\nu_1+\nu_2=a$, and $\nu_3=-b$. If $\eta_1>\eta_2=\eta_3$ then $\nu_1=a$ and $0\geq\nu_2\geq\nu_3\geq-b$, $\nu_2+\nu_3=-b$. }
\label{remark_subdif2}
\end{remark}

\subsection{First and second derivatives of eigenvalue functions}
\label{subsec_eigenprojection}

It is well-known that differentiability of eigenvalue functions depends on multiplicity of the eigenvalues. For example, the function $g$ is differentiable at $\mbf\eta$ with $\eta_1>\eta_2>\eta_3$ as follows from  Remark \ref{remark_subdif2}. Following \cite{NPO08, CaHo86}, we derive the first and second Fr\'echet derivatives of the eigenvalue functions using eigenprojections.  The derivative of function $F:\mathbb R^{3\times 3}_{sym}\rightarrow\mathbb R$ at $\mbf\eta$ is denoted as $\mathcal D F(\mbf\eta)$. Analogous notation, $\mathcal D \mbf F(\mbf\eta)$, is also used for tensor-valued function  $\mbf F:\mathbb R^{3\times 3}_{sym}\rightarrow\mathbb R^{3\times 3}_{sym}$. Further, it is worth mentioning that some derivatives introduced below cannot be extended on $\mathbb R^{3\times 3}$.

First, assume three distinct eigenvalues of $\mbf\eta$, i.e., $\eta_1> \eta_2> \eta_3$. Then one can introduce the eigenprojections $\mbf E_i:=\mbf E_i(\mbf\eta)$, $i=1,2,3$, of $\mbf\eta$ as follows:
\begin{equation}
\mbf E_i=\mathbf e_i\otimes\mathbf e_i=\frac{(\mbf\eta-\eta_j\mbf I)(\mbf\eta-\eta_k\mbf I)}{(\eta_i-\eta_j)(\eta_i-\eta_k)},\quad i\neq j\neq k\neq i,\;\; i=1,2,3.
\label{eigenprojection_case1}
\end{equation}
It holds:
\begin{equation}
\mbf\eta=\sum_{i=1}^3\eta_i\mbf E_i,\quad \sum_{i=1}^3\mbf E_i=\mbf I,
\label{eigen_prop}
\end{equation}
\begin{equation}
\mathcal D\omega_i(\mbf\eta)=\mbf E_i(\mbf\eta),\quad i=1,2,3,
\label{der_eigen_case1}
\end{equation}
\begin{equation}
\mathcal D\mbf E_i(\mbf\eta)=\frac{\mathcal D(\mbf\eta^2)-(\eta_j+\eta_k)\mathbb I-(2\eta_i-\eta_j-\eta_k)\mbf E_i\otimes\mbf E_i-(\eta_j-\eta_k)[\mbf E_j\otimes\mbf E_j-\mbf E_k\otimes\mbf E_k]}{(\eta_i-\eta_j)(\eta_i-\eta_k)},
\label{E_deriv_case1}
\end{equation}
for any $i=1,2,3,$ $i\neq j\neq k\neq i$, where the components of the fourth order tensors $\mathcal D(\mbf\eta^2)$ and $\mathbb I$ satisfy
$[\mathcal D(\mbf\eta^2)]_{ijkl}=\delta_{ik}[\mbf\eta]_{lj}+\delta_{jl}[\mbf\eta]_{ik}$ and $[\mathbb I]_{ijkl}=\delta_{ik}\delta_{jl}$, respectively\footnote{In \cite[Appendix A]{NPO08}, instead of $\mathcal D(\mbf\eta^2)$ and $\mathbb I$, their symmetric parts are introduced. For example, instead of $\mathbb I$, the tensor $\mathbb I_S$ with the components $[\mathbb I_S]_{ijkl}=\frac{1}{2}(\delta_{ik}\delta_{jl}+\delta_{il}\delta_{jk})$ is considered. One can easily check that $\mathbb I:\mbf\eta=\mathbb I_S:\mbf\eta=\mbf\eta$ for any $\mbf\eta\in\mathbb R^{3\times 3}_{sym}$. A similar identity also holds for $\mathcal D(\mbf\eta^2)$.}.
We use the notation $\mathbb E_i(\mbf\eta):=\mathcal D\mbf E_i(\mbf\eta)$,  $i=1,2,3$. 

Now, assume $\eta_1\geq\eta_2>\eta_3$. In this more general case, one can introduce the derivatives of $\omega_3$ and $\omega_{12}:=\omega_1+\omega_2$. From (\ref{eigenprojection_case1}), it is readily seen that the function $\mbf E_3$ can be continuously extended for $\mbf\eta$ satisfying $\eta_1=\eta_2$ unlike $\mbf E_1$ and $\mbf E_2$. Hence and from (\ref{eigen_prop}), (\ref{der_eigen_case1}), one can write:
\begin{equation}
\mathcal D\omega_3(\mbf\eta)=\mbf E_3(\mbf\eta),\quad \mathcal D\omega_{12}(\mbf\eta)=\mbf I-\mbf E_{3}(\mbf\eta)=: \mbf E_{12}(\mbf\eta).
\label{der_eigen_case2}
\end{equation}
To continuously extend the function $\mathbb E_3(\mbf\eta):=\mathcal D\mbf E_3(\mbf\eta)=-\mathcal D\mbf E_{12}(\mbf\eta)$, we use the equality
$$(\eta_1-\eta_2)(\mbf E_1\otimes\mbf E_1-\mbf E_2\otimes\mbf E_2)=(\mbf\eta-\eta_3\mbf E_3)\otimes\mbf E_{12}+\mbf E_{12}\otimes(\mbf\eta-\eta_3\mbf E_3)-(\eta_1+\eta_2)\mbf E_{12}\otimes\mbf E_{12}$$
and substitute it into (\ref{E_deriv_case1}) for $i=3$. We obtain
\begin{eqnarray}
\mathbb E_3(\mbf\eta)&=&\frac{\mathcal D(\mbf\eta^2)-(\eta_1+\eta_2)\mathbb I-[\mbf\eta\otimes\mbf E_{12}+\mbf E_{12}\otimes\mbf\eta]+(\eta_1+\eta_2)\mbf E_{12}\otimes\mbf E_{12}}{(\eta_3-\eta_1)(\eta_3-\eta_2)}+\nonumber\\[3pt]
&&+\frac{(\eta_1+\eta_2-2\eta_3)\mbf E_3\otimes\mbf E_3+\eta_3[\mbf E_{12}\otimes\mbf E_3+\mbf E_3\otimes\mbf E_{12}]}{(\eta_3-\eta_1)(\eta_3-\eta_2)}.
\label{E3_deriv_case2}
\end{eqnarray}
Clearly, (\ref{E3_deriv_case2}) is well-defined also for $\eta_1=\eta_2$. Notice that if $\eta_1=\eta_2>\eta_3$ then $\mbf \eta$ has only two eigenprojections: $\mbf E_{12}$ and $\mbf E_3$, and $\mbf\eta=\eta_1\mbf E_{12}+\eta_3\mbf E_{3}$. Conversely, if $\eta_1>\eta_2>\eta_3$, then $\mbf E_{12}=\mbf E_1+\mbf E_2$.

If $\eta_1> \eta_2\geq \eta_3$ then one can introduce the derivatives of the functions $\omega_1$, $\omega_{23}:=\omega_2+\omega_3$. Similarly as in the previous case, it holds:
\begin{equation}
\mathcal D\omega_1(\mbf\eta)=\mbf E_1(\mbf\eta),\quad \mathcal D\omega_{23}(\mbf\eta)=\mbf I-\mbf E_{1}(\mbf\eta)=: \mbf E_{23}(\mbf\eta),
\label{der_eigen_case3}
\end{equation}
\begin{eqnarray}
\mathbb E_1(\mbf\eta)=\mathcal D\mbf E_1(\mbf\eta)&=&\frac{\mathcal D(\mbf\eta^2)-(\eta_2+\eta_3)\mathbb I-[\mbf\eta\otimes\mbf E_{23}+\mbf E_{23}\otimes\mbf\eta]+(\eta_2+\eta_3)\mbf E_{23}\otimes\mbf E_{23}}{(\eta_1-\eta_2)(\eta_1-\eta_2)}+\nonumber\\[3pt]
&&+\frac{(\eta_2+\eta_3-2\eta_1)\mbf E_1\otimes\mbf E_1+\eta_1[\mbf E_{23}\otimes\mbf E_1+\mbf E_1\otimes\mbf E_{23}]}{(\eta_1-\eta_2)(\eta_1-\eta_3)}.
\label{E3_deriv_case3}
\end{eqnarray}
Notice that if $\eta_1>\eta_2=\eta_3$ then $\mbf \eta$ has only two eigenprojections: $\mbf E_{1}$ and $\mbf E_{23}$, and $\mbf\eta=\eta_1\mbf E_1+\eta_3\mbf E_{23}$. Conversely, if $\eta_1>\eta_2>\eta_3$, then $\mbf E_{23}=\mbf E_2+\mbf E_3$.

In the general case $\eta_1\geq \eta_2\geq \eta_3$, it holds that $\eta_1+\eta_2+\eta_3=\mbf\eta:\mbf I$ and thus
\begin{equation}
\mathcal D[\omega_1+\omega_2+\omega_3](\mbf\eta)=\mbf I.
\label{spectrum_case4}
\end{equation}
Notice that if $\eta_1=\eta_2=\eta_3$ then $\mbf \eta=\eta_1\mbf I$ has only one eigenprojection: $\mbf I$.

\begin{remark}
\emph{The mentioned derivatives can be found in simpler forms when plane strain assumptions are considered, see Appendix A of this paper.}
\end{remark}

%%%%%%%%%%%          Section 3        %%%%%%%%%%%%%%%%%%%%%%%%%%%%%%%%%%%%%%%%%%%%%%%%
\section{The Mohr-Coulomb constitutive problems}
\label{sec_model}

In this section, we introduce the Mohr-Coulomb constitutive initial value problem and its implicit Euler discretization. We use the model proposed in \cite{NPO08} containing the Mohr-Coulomb yield criterion, the  nonassociative plastic flow rule, and the nonlinear isotropic hardening.

\subsection{The initial value constitutive problem}

The initial value constitutive problem reads as:

\medskip\noindent
\textit{Given the history of the strain tensor $\mbf\varepsilon=\mbf\varepsilon(t)$, $t\in[0, t_{\max}]$, and the initial values
$\mbf{\varepsilon}^p(0)=\mbf{\varepsilon}^p_0, \; \bar\varepsilon^p(0)=\bar\varepsilon^p_0.$
Find $(\mbf\sigma(t),\mbf{\varepsilon}^p(t), \bar\varepsilon^p(t))$ such that 
\begin{equation}
\left.
\begin{array}{c}
\mbf\sigma=\mathbb D_e:(\mbf{\varepsilon}-\mbf{\varepsilon}^p),\;\; \kappa=H(\bar\varepsilon^p),\\[1mm]
\dot{\mbf{\varepsilon}^p}\in\dot\lambda\partial g(\mbf\sigma),\;\; \dot{\bar{\varepsilon}}^p=-\dot\lambda\frac{\partial f(\mbf\sigma,\kappa)}{\partial \kappa},\\[1mm]
\dot\lambda\geq0,\;\; f(\mbf\sigma,\kappa)\leq0,\;\; \dot\lambda f(\mbf\sigma,\kappa)=0.
\end{array}
\right\}
\label{CIVP_MC}
\end{equation}
hold for each instant $t\in[0,t_{\max}]$}.

\medskip\noindent
Here, $\mbf\sigma,\mbf{\varepsilon}^p, \bar\varepsilon^p,\lambda$ denote the Cauchy stress tensor, the plastic strain, the hardening variable, and the plastic multiplier, respectively. The dot symbol means the pseudo-time derivative of a quantity. The functions $f$ and $g$ represent the yield function and the plastic potential for the Mohr-Coulomb model, respectively. They are defined as follows:
\begin{eqnarray}
f(\mbf\sigma,\kappa)&=&(1+\sin\phi)\omega_1(\mbf\sigma)-(1-\sin\phi)\omega_3(\mbf\sigma)-2(c_0+\kappa)\cos\phi,\label{yield_function_simpl}\\
g(\mbf\sigma)&=&(1+\sin\psi)\omega_1(\mbf\sigma)-(1-\sin\psi)\omega_3(\mbf\sigma),
\label{potential_function}
\end{eqnarray}
where $\omega_1$ and $\omega_3$ are the maximal and minimal eigenvalue functions introduced in Section \ref{sec_spectrum}, and
he material parameters $c_0>0$, $\phi,\psi\in(0,\pi/2)$ represent the initial cohesion, the friction angle, and the dilatancy angle, respectively. Notice that $f,g$ are convex functions with respect to the stress variable. Recall that the function $g$ was already introduced in Section \ref{subsec_subdif} for the choice
\begin{equation}
a:=1+\sin\psi,\quad b:=1-\sin\psi
\label{def_ab}
\end{equation}
and thus one can define $\partial g(\mbf\sigma)$ using Lemma \ref{lem_subdif}.
Clearly, $\partial f(\mbf\sigma,\kappa)/\partial \kappa=-2\cos\phi$.

Further, the fourth order tensor $\mathbb D_e$ represents linear isotropic elastic law:
\begin{equation}
\mbf\sigma=\mathbb D_e:\mbf{\varepsilon}^e=\frac{1}{3}(3K-2G)(\mbf I:\mbf{\varepsilon}^e)\mbf I+2G\mbf{\varepsilon}^e,\quad \mathbb D_e=\frac{1}{3}(3K-2G)\mbf I\otimes\mbf I+2G\mathbb I,
\label{elastic_law}
\end{equation}
where $\mbf\varepsilon^e=\mbf\varepsilon-\mbf\varepsilon^p$ is the elastic part of the strain tensor and $K,G>0$ denotes the bulk, and shear moduli, respectively. 

Finally, we let the function $H$ representing the non-linear isotropic hardening in an abstract form and assume that
it is a nondecreasing, continuous, and piecewise smooth function satisfying $H(0)=0$. 

It is worth mentioning that the value $t_{max}$ need not be always known, see Section \ref{sec_realization}.

\subsection{The discretized constitutive problem}

Let $0=t_0<t_1<\ldots<t_k<\ldots<t_N=t_{\max}$ be a partition of the interval $[0,t_{\max}]$ and denote $\mbf{\sigma}_k:=\mbf{\sigma}(t_k)$, $\mbf\varepsilon_k:=\mbf\varepsilon(t_k)$, $\mbf\varepsilon^p_k:=\mbf\varepsilon^p(t_k)$, $\bar\varepsilon^p_k:=\bar\varepsilon^p(t_k)$, $\bar{\varepsilon}^{p,tr}_k:=\bar{\varepsilon}^p(t_{k-1})$, $\mbf\varepsilon^{tr}_k:=\mbf\varepsilon(t_{k})-\mbf{\varepsilon}^p(t_{k-1})$, and $\mbf{\sigma}^{tr}_k:=\mathbb D_e:\mbf\varepsilon^{tr}_k$. Here, the superscript $tr$ is the standard notation for the so-called trial variables (see, e.g., \cite{NPO08}) which are known. If it is clear that the step $k$ is fixed then we will omit the subscript $k$ and write $\mbf{\sigma}$, $\mbf\varepsilon$, $\mbf\varepsilon^p$, $\bar\varepsilon^p$, $\bar{\varepsilon}^{p,tr}$, $\mbf\varepsilon^{tr}$, and $\mbf{\sigma}^{tr}$ to simplify the notation. The $k$-th step of the incremental constitutive problem discretized by the implicit Euler method reads as: 

\medskip\noindent
{\it Given $\mbf{\sigma}^{tr}$ and $\bar{\varepsilon}^{p,tr}$. Find $\mbf{\sigma}$,  $\bar{\varepsilon}^p$, and $\triangle\lambda$ satisfying:}
\begin{equation}
\left.
\begin{array}{c}
\mbf{\sigma}=\mbf{\sigma}^{tr}-\triangle\lambda\mathbb D_e:\mbf\nu,\quad \mbf\nu\in\partial g(\mbf\sigma),\\[1mm]
\bar{\varepsilon}^p=\bar{\varepsilon}^{p,tr}+\triangle\lambda (2\cos\phi),\\[1mm]
\triangle\lambda\geq0,\quad f(\mbf{\sigma},H(\bar\varepsilon^p))\leq0,\quad \triangle\lambda f(\mbf{\sigma},H(\bar\varepsilon^p))=0.
\end{array}
\right\}
\label{k_step_problem}
\end{equation}
Unlike problem (\ref{CIVP_MC}), the unknown $\mbf\varepsilon^p$ is not introduced in (\ref{k_step_problem}). It can be simply computed from the formula $\mbf{\varepsilon}^p(t_k)=\mbf{\varepsilon}(t_k)-\mathbb D_e^{-1}:\mbf{\sigma}(t_k)$ and used as the input parameter for the next step.

%********************** Section 4**********************************************************************
\section{Solution of the discretized constitutive problem}
\label{sec_time_discret}

The aim of this section is to derive an improved solution scheme to problem (\ref{k_step_problem}).  The solution scheme builds on the standard {\it elastic predictor - plastic corrector method} and its improvement is based on the form of $\partial g(\mbf\sigma)$ introduced in Lemma \ref{lem_subdif}. Within the {\it elastic prediction}, we assume $\triangle\lambda=0$. Then, it is readily seen that the triple
\begin{equation}
\mbf{\sigma}=\mbf{\sigma}^{tr},\quad \bar{\varepsilon}^p=\bar{\varepsilon}^{p,tr}, \quad \triangle\lambda=0
\label{elast_solution}
\end{equation}
is the solution to (\ref{k_step_problem}) under the condition
\begin{equation}
f(\mbf{\sigma}^{tr},H(\bar{\varepsilon}^{p,tr}))\leq0.
\label{trial_admissibility}
\end{equation}
The {\it plastic correction} happens when $\triangle\lambda>0$. Then the unknown generalized stress $(\mbf\sigma,H(\bar{\varepsilon}^p))$ lies on the yield surface and thus the corresponding plastic correction problem reads as:
{\it Given $\mbf{\sigma}^{tr}$ and $\bar{\varepsilon}^{p,tr}$. Find $\mbf{\sigma}$,  $\bar{\varepsilon}^p$, and $\triangle\lambda>0$ satisfying:}
\begin{equation}
\left.
\begin{array}{c}
\mbf{\sigma}=\mbf{\sigma}^{tr}-\triangle\lambda\mathbb D_e:\mbf\nu,\quad \mbf\nu\in\partial g(\mbf\sigma),\\[1mm]
\bar{\varepsilon}^p=\bar{\varepsilon}^{p,tr}+\triangle\lambda (2\cos\phi),\\[1mm]
f(\mbf{\sigma},H(\bar\varepsilon^p))=0.
\end{array}
\right\}
\label{correction_problem}
\end{equation}
The solution scheme to problem (\ref{correction_problem}) is usually called the {\it implicit return-mapping scheme}. Since its derivation is technically complicated, we divide the rest of this section into several subsections for easier orientation in the text. In Section \ref{subsec_corrections}, problem (\ref{correction_problem}) is reduced and written in terms of principal stresses. In parallel Sections \ref{subsec_smooth}-\ref{subsec_apex}, we introduce solution schemes for returns to the smooth portion, to the ``left" edge, to the ``right" edge, and to the apex of the pyramidal yield surface, respectively.  In Section \ref{subsec_lambda}, we derive a nonlinear equation for the unknown plastic multiplier. This equation is common for all types of the return and has the unique solution. Hence, we derive: existence and uniqueness of problems (\ref{k_step_problem}) and (\ref{correction_problem}), a priori decision criteria for the return types, and other useful results describing a dependence of the unknown stress tensor on the trial stress tensor.
 
\subsection{Plastic correction problem in terms of principal stresses}
\label{subsec_corrections}

First, we reduce problem (\ref{correction_problem}) using the spectral decomposition of $\mbf \sigma$ (see Section \ref{sec_spectrum}):
\begin{equation}
\mbf\sigma=\sum_{i=1}^3\sigma_i\mathbf e_i\otimes\mathbf e_i,\quad \sigma_1\geq\sigma_2\geq\sigma_3,\quad (\mathbf e_1, \mathbf e_2, \mathbf e_3)\in V(\mbf\sigma), \quad \sigma_i:=\omega_i(\mbf\sigma),\;i=1,2,3.
\label{spectrum_sigma}
\end{equation}
From the definition of $f$ introduced in Section \ref{sec_model}, it is easy to see that the equation (\ref{correction_problem})$_3$ can be written only in terms the principal stresses $\sigma_1,\sigma_2,\sigma_3$ instead of the whole stress tensor $\mbf\sigma$. To re-formulate (\ref{correction_problem})$_1$, we use Lemma \ref{lem_subdif} and (\ref{def_ab}): there exists $(\mathbf e_1, \mathbf e_2, \mathbf e_3)\in V(\mbf\sigma)$ such that $\mbf\nu=\sum_{i=1}^3 \nu_i\mathbf e_i\otimes\mathbf e_i$, where
\begin{equation}
\left.
\begin{array}{c}
1+\sin\psi\geq \nu_1\geq \nu_2\geq \nu_3\geq-1+\sin\psi, \quad \nu_1+ \nu_2+ \nu_3=2\sin\psi,\\[1 mm]
( \nu_1-1-\sin\psi)(\sigma_1-\sigma_2)=0,\quad ( \nu_3+1-\sin\psi)(\sigma_2-\sigma_3)=0.
\end{array}
\right\}
\label{prop_n_i}
\end{equation}
Since $\mbf I=\sum_{i=1}^3\mathbf e_i\otimes\mathbf e_i$, (\ref{elastic_law}) implies
\begin{equation}
\mathbb D_e:\mbf\nu=\sum_{i=1}^3 \left[\frac{2}{3}(3K-2G)\sin\psi+2G\nu_i\right]\mathbf e_i\otimes\mathbf e_i.
\label{Den_spectrum}
\end{equation}
Then one can substitute (\ref{spectrum_sigma}) and (\ref{Den_spectrum}) to (\ref{correction_problem})$_1$:
\begin{equation}
\mbf{\sigma}^{tr}=\mbf\sigma+\triangle\lambda\mathbb D_e:\mbf\nu=\sum_{i=1}^3\sigma_i^{tr}\mathbf e_i\otimes\mathbf e_i,\;\;\mbox{where}\quad \sigma_i^{tr}=\sigma_i+\triangle\lambda\left[\frac{2}{3}(3K-2G)\sin\psi+2G\nu_i\right].
\label{flow_eigen}
\end{equation}
Notice that (\ref{flow_eigen})$_1$ defines the spectral decomposition of $\mbf\sigma^{tr}$.
Since $\sigma_1\geq\sigma_2\geq\sigma_3$ and $\nu_1\geq\nu_2\geq\nu_3$, we have:
\begin{itemize}
\item[$(i)$] $\sigma_1^{tr}\geq\sigma_2^{tr}\geq\sigma_3^{tr}$; 
\item[$(ii)$] if $\sigma_i^{tr}=\sigma_j^{tr}$ then $\sigma_i=\sigma_j$, $\nu_i=\nu_j$.
\end{itemize}
From $(i)$, it follows that the eigenvalues $\sigma_1^{tr},\sigma_2^{tr},\sigma_3^{tr}$ are ordered and thus uniquely determined using the eigenvalue functions: $\sigma_i^{tr}=\omega_i(\mbf\sigma^{tr})$, $i=1,2,3$. From $(ii)$, we conclude that $\mbf\sigma=\sum_{i=1}^3\sigma_i\mathbf e_i^{tr}\otimes\mathbf e_i^{tr}$, $\mbf\nu=\sum_{i=1}^3\nu_i\mathbf e_i^{tr}\otimes\mathbf e_i^{tr}$ for any $(\mathbf e_1^{tr},\mathbf e_2^{tr},\mathbf e_3^{tr})\in V(\mbf\sigma^{tr})$. The following lemma summarizes the proven results.

\begin{lem}
Let $(\mbf{\sigma}, \bar{\varepsilon}^p,\triangle\lambda)$, $\triangle\lambda>0$ be a solution to (\ref{correction_problem}) for given $\mbf{\sigma}^{tr}$ and $\bar{\varepsilon}^{p,tr}$. Let $\sigma_i$, $\sigma_i^{tr}$, $i=1,2,3$, be the ordered eigenvalues of $\mbf{\sigma}$ and $\mbf{\sigma}^{tr}$, respectively. Then $(\sigma_1,\sigma_2,\sigma_3, \bar{\varepsilon}^p,\triangle\lambda)$ is a solution to:
\begin{equation}
\left.
\begin{array}{c}
\sigma_i=\sigma_i^{tr}-\triangle\lambda\left[\frac{2}{3}(3K-2G)\sin\psi+2G\nu_i\right],\quad i=1,2,3,\\[1mm]
\bar{\varepsilon}^p=\bar{\varepsilon}^{p,tr}+\triangle\lambda (2\cos\phi),\\[1mm]
(1+\sin\phi)\sigma_1-(1-\sin\phi)\sigma_3-2(c_0+H(\bar\varepsilon^p))\cos\phi=0,
\end{array}
\right\}
\label{correction_problem2}
\end{equation}
where $\nu_1,\nu_2,\nu_3$ satisfy (\ref{prop_n_i})
Conversely, if $(\sigma_1,\sigma_2,\sigma_3, \bar{\varepsilon}^p,\triangle\lambda)$, $\triangle\lambda>0$ is a solution to (\ref{correction_problem2}) then $(\mbf{\sigma}, \bar{\varepsilon}^p,\triangle\lambda)$ solves (\ref{correction_problem}), where $\mbf\sigma=\sum_{i=1}^3\sigma_i\mathbf e_i^{tr}\otimes\mathbf e_i^{tr}$, $(\mathbf e_1^{tr},\mathbf e_2^{tr},\mathbf e_3^{tr})\in V(\mbf\sigma^{tr})$.
\label{lem_reduction1}
\end{lem}

To be in accordance with problems (\ref{k_step_problem}) and (\ref{correction_problem}), we do not include $\nu_1,\nu_2,\nu_3$ to the list of unknowns. From (\ref{prop_n_i}), it follows that
the values of $\nu_1,\nu_2,\nu_3$ can be specified depending on multiplicity of $\sigma_1,\sigma_2,\sigma_3$, similarly as in Remark \ref{remark_subdif2}. Therefore, we will distinguish below four types of the return on the yield surface: the return to the smooth portion ($\sigma_1>\sigma_2>\sigma_3$), the return to the left edge ($\sigma_1=\sigma_2>\sigma_3$), the return to the right edge ($\sigma_1>\sigma_2=\sigma_3$) and the return to the apex ($\sigma_1=\sigma_2=\sigma_3$). This terminology follows from \cite{NPO08}, another one is used, e.g., in \cite{LR96}. Within the below introduced notation, we will use the subscripts $s$, $l$, $r$, $a$ to distinguish the return type and the superscript ``$tr$" to emphasize a known quantity depending only on the trial variables.

\subsection{The return to the smooth portion}
\label{subsec_smooth}

Assume $\sigma_1>\sigma_2>\sigma_3$. Then $\nu_1=1+\sin\psi$, $\nu_2=0$, $\nu_3=-(1-\sin\psi)$ and (\ref{correction_problem2})$_1$ reads as:
\begin{eqnarray}
\sigma_1&=&\sigma_1^{tr}-\triangle\lambda\left[\frac{2}{3}(3K-2G)\sin\psi+2G(1+\sin\psi)\right],\label{flow11}\\
\sigma_2&=&\sigma_2^{tr}-\triangle\lambda\left[\frac{2}{3}(3K-2G)\sin\psi\right],\label{flow12}\\
\sigma_3&=&\sigma_3^{tr}-\triangle\lambda\left[\frac{2}{3}(3K-2G)\sin\psi-2G(1-\sin\psi)\right].\label{flow13}
\end{eqnarray}
Consequently, one can substitute (\ref{flow11}), (\ref{flow13}), and (\ref{correction_problem2})$_2$ to (\ref{correction_problem2})$_3$. This leads to the equation $q_s^{tr}(\triangle\lambda)=0$, where 
\begin{eqnarray}
q^{tr}_s(\gamma)&=&(1+\sin\phi)\sigma_1^{tr}-(1-\sin\phi)\sigma_3^{tr}-2\left[c_0+H\left(\bar{\varepsilon}^{p,tr}+\gamma (2\cos\phi)\right)\right]\cos\phi\nonumber\\
&&-\gamma\left[\frac{4}{3}(3K-2G)\sin\psi\sin\phi+ 4G(1+\sin\psi\sin\phi)\right]. \label{q1}
\end{eqnarray}
Further, from (\ref{flow11})-(\ref{flow13}), two additional important consequence follow:
\begin{itemize}
\item $\sigma_1^{tr}>\sigma_2^{tr}>\sigma_3^{tr}$,
\item $\triangle\lambda\in C^{tr}_s:=\{\gamma\in(0,+\infty)\ |\; \gamma<\min\{\gamma^{tr}_{s,l},\gamma^{tr}_{s,r}\}\}$, where
\end{itemize}
\begin{equation}
\gamma^{tr}_{s,l}:=\frac{\sigma_1^{tr}-\sigma_2^{tr}}{2G(1+\sin\psi)}\geq0,\quad\gamma^{tr}_{s,r}:=\frac{\sigma_2^{tr}-\sigma_3^{tr}}{2G(1-\sin\psi)}\geq0.
\label{gamma1}
\end{equation}

\subsection{The return to the left edge} 

Assume $\sigma_1=\sigma_2>\sigma_3$. Then $\nu_3=-(1-\sin\psi)$, $\nu_1+\nu_2=1+\sin\psi$, and $1+\sin\psi\geq \nu_1\geq \nu_2\geq0$ implying $\nu_1-\nu_2\leq1+\sin\psi$.  Consequently, (\ref{correction_problem2})$_1$ yields:
\begin{eqnarray}
\frac{1}{2}(\sigma_1+\sigma_2)=\sigma_1&=&\frac{1}{2}(\sigma_1^{tr}+\sigma_2^{tr})-\triangle\lambda\left[\frac{2}{3}(3K-2G)\sin\psi+G(1+\sin\psi)\right],\label{flow212}\\
\sigma_3&=&\sigma_3^{tr}-\triangle\lambda\left[\frac{2}{3}(3K-2G)\sin\psi-2G(1-\sin\psi)\right],\label{flow23}
\end{eqnarray}
and
\begin{equation}
0=\sigma_1-\sigma_2=\sigma_1^{tr}-\sigma_2^{tr}-\triangle\lambda[2G(\nu_1-\nu_2)]\geq\sigma_1^{tr}-\sigma_2^{tr}-\triangle\lambda[2G(1+\sin\psi)].
\label{est_left}
\end{equation}
After substitution (\ref{flow212}), (\ref{flow23}), and (\ref{correction_problem2})$_2$ to (\ref{correction_problem2})$_3$, we arrive at $q_l^{tr}(\triangle\lambda)=0$, where 
\begin{eqnarray}
q^{tr}_l(\gamma)&=&\frac{1}{2}(1+\sin\phi)(\sigma_1^{tr}+\sigma_2^{tr})-(1-\sin\phi)\sigma_3^{tr}-2\left[c_0+H\left(\bar{\varepsilon}^{p,tr}+\gamma (2\cos\phi)\right)\right]\cos\phi-\nonumber\\
&&\gamma\left[\frac{4}{3}(3K-2G)\sin\psi\sin\phi+ G(1+\sin\psi)(1+\sin\phi)+ 2G(1-\sin\psi)(1-\sin\phi)\right].\qquad
\label{q2}
\end{eqnarray}
Further, from (\ref{flow212})-(\ref{q2}), three additional important consequences follow:
\begin{itemize}
\item $\sigma_2^{tr}>\sigma_3^{tr}$,
\item $\sigma_1$, $\sigma_3$, $\triangle\lambda$ depend on $\sigma^{tr}_1$, $\sigma^{tr}_2$ only through $\sigma^{tr}_1+\sigma^{tr}_2$,
\item $\triangle\lambda\in C^{tr}_l:=\{\gamma\in(0,+\infty)\ |\; \gamma^{tr}_{s,l}\leq\gamma<\gamma^{tr}_{l,a}\}$, where
\begin{equation}
\gamma^{tr}_{l,a}=\frac{\sigma_1^{tr}+\sigma_2^{tr}-2\sigma_3^{tr}}{2G(3-\sin\psi)}=\frac{1+\sin\psi}{3-\sin\psi}\gamma^{tr}_{s,l}+\left(1-\frac{1+\sin\psi}{3-\sin\psi}\right)\gamma^{tr}_{s,r}\geq0
\label{gamma2}
\end{equation}
and $\gamma^{tr}_{s,l}$, $\gamma^{tr}_{s,r}$ are the same as in (\ref{gamma1}). Notice that $\gamma^{tr}_{s,l}<\gamma^{tr}_{l,a}<\gamma^{tr}_{s,r}$ in this case.
\end{itemize}

\subsection{The return to the right edge} 

Assume $\sigma_1>\sigma_2=\sigma_3$. 
Then $\nu_1=1+\sin\psi$, $\nu_2+\nu_3=-1+\sin\psi$, and $0\geq \nu_2\geq \nu_3\geq -1+\sin\psi$ implying $\nu_2-\nu_3\leq1-\sin\psi$.  Consequently, (\ref{correction_problem2})$_1$ yields:
\begin{eqnarray}
\sigma_1&=&\sigma_1^{tr}-\triangle\lambda\left[\frac{2}{3}(3K-2G)\sin\psi+2G(1+\sin\psi)\right],\label{flow31}\\
\frac{1}{2}(\sigma_2+\sigma_3)=\sigma_3&=&\frac{1}{2}(\sigma_2^{tr}+\sigma_3^{tr})-\triangle\lambda\left[\frac{2}{3}(3K-2G)\sin\psi- G(1-\sin\psi)\right].\label{flow323}
\end{eqnarray}
and
\begin{equation}
0=\sigma_2-\sigma_3=\sigma_2^{tr}-\sigma_3^{tr}-\triangle\lambda[2G(\nu_2-\nu_3)]\geq\sigma_2^{tr}-\sigma_3^{tr}-\triangle\lambda[2G(1-\sin\psi)].
\label{est_right}
\end{equation}
After substitution (\ref{flow31}), (\ref{flow323}), and (\ref{correction_problem2})$_2$ into (\ref{correction_problem2})$_3$, we arrive at $q_r^{tr}(\triangle\lambda)=0$, where 
\begin{eqnarray}
q^{tr}_r(\gamma)&=&(1+\sin\phi)\sigma_1^{tr}-\frac{1}{2}(1-\sin\phi)(\sigma_2^{tr}+\sigma_3^{tr})-2\left[c_0+H\left(\bar{\varepsilon}^{p,tr}+\gamma (2\cos\phi)\right)\right]\cos\phi-\nonumber\\
&&\gamma\left[\frac{4}{3}(3K-2G)\sin\psi\sin\phi+ 2G(1+\sin\psi)(1+\sin\phi)+ G(1-\sin\psi)(1-\sin\phi)\right].\qquad
\label{q3}
\end{eqnarray}
Further, from (\ref{flow31})-(\ref{q3}),  three additional important consequences follow:
\begin{itemize}
\item $\sigma_1^{tr}>\sigma_2^{tr}\geq\sigma_3^{tr}$,
\item $\sigma_1$, $\sigma_3$, $\triangle\lambda$ depend on $\sigma^{tr}_2$, $\sigma^{tr}_3$ only through $\sigma^{tr}_2+\sigma^{tr}_3$,
\item $\triangle\lambda\in C^{tr}_r:=\{\gamma\in(0,+\infty)\ |\; \gamma^{tr}_{s,r}\leq\gamma<\gamma^{tr}_{r,a}\}$, where
\begin{equation}
\gamma^{tr}_{r,a}=\frac{2\sigma_1^{tr}-\sigma_2^{tr}-\sigma_3^{tr}}{2G(3+\sin\psi)}=\frac{1-\sin\psi}{3+\sin\psi}\gamma^{tr}_{s,r}+\left(1-\frac{1-\sin\psi}{3+\sin\psi}\right)\gamma^{tr}_{s,l}\geq0.
\label{gamma3}
\end{equation}
and $\gamma^{tr}_{s,l}$, $\gamma^{tr}_{s,r}$ are the same as in (\ref{gamma1}). Notice that $\gamma^{tr}_{s,r}<\gamma^{tr}_{r,a}<\gamma^{tr}_{s,l}$ in this case.
\end{itemize}

\subsection{The return to the apex} 
\label{subsec_apex}

Assume $\sigma_1=\sigma_2=\sigma_3$. Then $ \nu_1+ \nu_2+ \nu_3=2\sin\psi$ and $1+\sin\psi\geq \nu_1\geq \nu_2\geq \nu_3\geq-1+\sin\psi$ implying $2\nu_1-\nu_2-\nu_3\leq3+\sin\psi$, $\nu_1+\nu_2-2\nu_3\leq3-\sin\psi$. Consequently, (\ref{correction_problem2})$_1$ yields:
\begin{equation}
\sigma_1=\frac{1}{3}(\sigma_1+\sigma_2+\sigma_3)=\frac{1}{3}(\sigma_1^{tr}+\sigma_2^{tr}+\sigma_3^{tr})-\triangle\lambda[2K\sin\psi]\quad\label{flow4123}
\end{equation}
and
\begin{eqnarray}
0&=&2\sigma_1-\sigma_2-\sigma_3\geq2\sigma_1^{tr}-\sigma_2^{tr}-\sigma_3^{tr}-\triangle\lambda[2G(3+\sin\psi)],\label{est_apex1}\\
0&=&\sigma_1+\sigma_2-2\sigma_3\geq\sigma_1^{tr}+\sigma_2^{tr}-2\sigma_3^{tr}-\triangle\lambda[2G(3-\sin\psi)].\label{est_apex2}
\end{eqnarray}
After substitution (\ref{flow4123}) and (\ref{correction_problem2})$_2$ into (\ref{correction_problem2})$_3$, we arrive at $q_a^{tr}(\triangle\lambda)=0$, where 
\begin{equation}
q^{tr}_a(\gamma)=\frac{2}{3}(\sigma_1^{tr}+\sigma_2^{tr}+\sigma_3^{tr})\sin\phi-2\left[c_0+H\left(\bar{\varepsilon}^{p,tr}+\gamma (2\cos\phi)\right)\right]\cos\phi-\gamma[4K\sin\psi\sin\phi]. \label{q4}
\end{equation}
Further, from (\ref{flow4123})-(\ref{q4}),  two additional important consequences follow:
\begin{itemize}
\item $\sigma_1$, $\triangle\lambda$ depend on $\sigma^{tr}_1$, $\sigma^{tr}_2$, $\sigma^{tr}_3$ only through $\sigma^{tr}_1+\sigma^{tr}_2+\sigma^{tr}_3$,
\item $\triangle\lambda\in C^{tr}_a,\quad C^{tr}_a:=\{\gamma\in(0,+\infty)\ |\; \gamma\geq\max\{\gamma^{tr}_{l,a},\gamma^{tr}_{r,a}\}\},$
where $\gamma^{tr}_{l,a}$, $\gamma^{tr}_{r,a}$ are the same as in (\ref{gamma2}), (\ref{gamma3}), respectively. 
\end{itemize}

\subsection{Solvability analysis and a priori decision criteria}
\label{subsec_lambda}

In parallel Sections \ref{subsec_smooth}-\ref{subsec_apex}, the solution schemes for the investigated return types were introduced. Similar schemes are also known from literature (see,  e.g., \cite[Section 8]{NPO08}) and their solutions are candidates on the solution to problem (\ref{correction_problem2}). This current approach is based on a blind guesswork since the position of the stress tensor on the yield surface is not a priori known. However at the ends of Sections \ref{subsec_smooth}-\ref{subsec_apex}, we also derived some additional results following from (\ref{prop_n_i}), i.e., from the knowledge of $\partial g(\mbf\sigma)$. These results enable to improve the solution scheme to problem (\ref{correction_problem2}).  First, we use the sets $C^{tr}_s$, $C^{tr}_l$, $C^{tr}_r$, $C^{tr}_a$, the values $\gamma^{tr}_{s,l}$, $\gamma^{tr}_{s,r}$, $\gamma^{tr}_{l,a}$, $\gamma^{tr}_{r,a}$, and the equations $q_s^{tr}(\triangle\lambda)=0$, $q_l^{tr}(\triangle\lambda)=0$, $q_r^{tr}(\triangle\lambda)=0$, $q_a^{tr}(\triangle\lambda)=0$ introduced above to find a unique nonlinear equation for the unknown plastic multiplier.
\begin{lem}
There exists a unique function $q^{tr}:\mathbb R_+\rightarrow\mathbb R$ satisfying:
\begin{itemize}
\item[$(i)$] $q^{tr}|_{C^{tr}_s}=q^{tr}_s$, $q^{tr}|_{C^{tr}_l}=q^{tr}_l$, $q^{tr}|_{C^{tr}_r}=q^{tr}_r$, $q^{tr}|_{C^{tr}_a}=q^{tr}_a$.
\item[$(ii)$] $q^{tr}$ is continuous, piecewise smooth and decreasing in $\mathbb R_+$.
\item[$(iii)$] $q^{tr}(0)=f(\mbf{\sigma}^{tr},H(\bar{\varepsilon}^{p,tr}))$.
\item[$(iv)$] $q^{tr}(\gamma)\rightarrow-\infty$ as $\gamma\rightarrow+\infty$.
\end{itemize}
\label{lem_crucial}
\end{lem}
\begin{proof}
Notice that the values $\gamma^{tr}_{s,l}$, $\gamma^{tr}_{s,r}$, $\gamma^{tr}_{l,a}$, and $\gamma^{tr}_{r,a}$ are nonnegative and a priori known. Moreover, from (\ref{gamma1}), (\ref{gamma2}) and (\ref{gamma3}), it follows that only two ordering of these values are possible: either $\gamma^{tr}_{s,l}\leq\gamma^{tr}_{l,a}\leq\gamma^{tr}_{r,a}\leq\gamma^{tr}_{s,r}$ or $\gamma^{tr}_{s,r}\leq\gamma^{tr}_{r,a}\leq\gamma^{tr}_{l,a}\leq\gamma^{tr}_{s,l}$. 

First, assume $\gamma^{tr}_{s,l}\leq\gamma^{tr}_{s,r}$. Then $C^{tr}_s=(0,\gamma^{tr}_{s,l})$, $C^{tr}_l=[\gamma^{tr}_{s,l},\gamma^{tr}_{l,a})$,  $C^{tr}_r=\emptyset$, and $C^{tr}_a=[\gamma^{tr}_{l,a},+\infty)$. Define the function 
\begin{eqnarray}
q^{tr}(\gamma)&=&(1+\sin\phi)\sigma_1^{tr}-(1-\sin\phi)\sigma_3^{tr}-\gamma\left[\frac{4}{3}(3K-2G)\sin\psi\sin\phi+4G(1+\sin\psi\sin\phi)\right]\nonumber\\
&&+G(1+\sin\psi)(1+\sin\phi)(\gamma-\gamma^{tr}_{s,l})^++\frac{1}{3}G(3-\sin\psi)(3-\sin\phi)(\gamma-\gamma^{tr}_{l,a})^+\nonumber\\
&&-2\left[c_0+H\left(\bar{\varepsilon}^{p,tr}+\gamma (2\cos\phi)\right)\right]\cos\phi,\quad\gamma\in(0,+\infty),
\label{q1_tr}
\end{eqnarray}
where $(.)^+$ denotes a positive part of a function. It is easy to verify that  $q^{tr}$ has the required properties under the assumptions on $H$ from Section \ref{sec_model}.

Secondly, assume  $\gamma^{tr}_{s,r}\leq\gamma^{tr}_{s,l}$. Then $C^{tr}_s=(0,\gamma^{tr}_{s,r})$,  $C^{tr}_l=\emptyset$, $C^{tr}_r=[\gamma^{tr}_{s,r},\gamma^{tr}_{r,a})$,  $C^{tr}_a=[\gamma^{tr}_{r,a},+\infty)$ and the function $q^{tr}$ with the required properties is defined as:
\begin{eqnarray}
q^{tr}(\gamma)&:=&(1+\sin\phi)\sigma_1^{tr}-(1-\sin\phi)\sigma_3^{tr}-\gamma\left[\frac{4}{3}(3K-2G)\sin\psi\sin\phi+4G(1+\sin\psi\sin\phi)\right]\nonumber\\
&&+G(1-\sin\psi)(1-\sin\phi)(\gamma-\gamma^{tr}_{s,r})^++\frac{1}{3}G(3+\sin\psi)(3+\sin\phi)(\gamma-\gamma^{tr}_{r,a})^+\nonumber\\
&&-2\left[c_0+H\left(\bar{\varepsilon}^{p,tr}+\gamma (2\cos\phi)\right)\right]\cos\phi.
\label{q2_tr}
\end{eqnarray}
\end{proof}

\begin{remark}
\emph{Notice that formulas (\ref{q1_tr}) and (\ref{q2_tr}) coincide for $\gamma^{tr}_{s,l}=\gamma^{tr}_{s,r}$. Hence, the function $q(\gamma; \sigma_1^{tr}, \sigma_1^{tr}, \sigma_1^{tr},\bar\varepsilon^{p,tr})=q^{tr}(\gamma)$ is continuous and piecewise smooth with respect to the trial variables. }
\label{rem_semismooth}
\end{remark}

Lemmas \ref{lem_crucial}, \ref{lem_reduction1} and (\ref{trial_admissibility}), (\ref{elast_solution}) imply the following main results.

\begin{theorem}
Let $q^{tr}(0)=f(\mbf{\sigma}^{tr},H(\bar{\varepsilon}^{p,tr}))\geq 0$. Then the equation $q^{tr}(\triangle\lambda)=0$ has a unique solution in $\mathbb R_+$. The solution vanishes if and only if $f(\mbf{\sigma}^{tr},H(\bar{\varepsilon}^{p,tr}))=0$. Moreover, if there are $\gamma_1,\gamma_2\geq0$ such that $\gamma_1<\gamma_2$, $q^{tr}(\gamma_1)>0$, and $q^{tr}(\gamma_2)<0$, then $\triangle\lambda\in(\gamma_1,\gamma_2)$. 
\end{theorem}

\begin{theorem}
Let $f(\mbf{\sigma}^{tr},H(\bar{\varepsilon}^{p,tr}))>0$. Then problems (\ref{correction_problem}) and (\ref{correction_problem2}) have a unique solution. The solution components to problem (\ref{correction_problem2}) can be found in the following way:
\begin{enumerate}
\item Let $q^{tr}_s(\min\{\gamma^{tr}_{s,l},\gamma^{tr}_{s,r}\})< 0$. Then $\triangle\lambda\in C^{tr}_s$ is the unique solution to $q^{tr}_s(\triangle\lambda)=0$ and $\sigma_1>\sigma_2>\sigma_3$ can be computed from  (\ref{flow11})-(\ref{flow13}). Moreover, $\sigma_1^{tr}>\sigma_2^{tr}>\sigma_3^{tr}$.
\item  Let $\gamma^{tr}_{s,l}<\gamma^{tr}_{l,a}$, $q^{tr}_l(\gamma^{tr}_{s,l})\geq 0$ and $q^{tr}_l(\gamma^{tr}_{l,a})< 0$. Then $\triangle\lambda\in C^{tr}_l$ is the unique solution to $q^{tr}_l(\triangle\lambda)=0$ and $\sigma_1=\sigma_2>\sigma_3$ can be computed from (\ref{flow212}),  and (\ref{flow23}). Moreover, $\sigma_2^{tr}>\sigma_3^{tr}$ and $\triangle\lambda$, $\sigma_1$, $\sigma_3$ depend on $\sigma^{tr}_1$, $\sigma^{tr}_2$ only through $\sigma^{tr}_1+\sigma^{tr}_2$.
\item Let $\gamma^{tr}_{s,r}<\gamma^{tr}_{r,a}$, $q^{tr}_r(\gamma^{tr}_{s,r})\geq 0$ and $q^{tr}_r(\gamma^{tr}_{r,a})< 0$. Then $\triangle\lambda\in C^{tr}_r$ is the unique solution to $q^{tr}_r(\triangle\lambda)=0$ and $\sigma_1>\sigma_2=\sigma_3$ can be computed from (\ref{flow31}),  and (\ref{flow323}). Moreover, $\sigma_1^{tr}>\sigma_2^{tr}$ and $\triangle\lambda$, $\sigma_1$, $\sigma_3$ depend on $\sigma^{tr}_2$, $\sigma^{tr}_3$ only through $\sigma^{tr}_2+\sigma^{tr}_3$.
\item  Let $q^{tr}_a(\max\{\gamma^{tr}_{l,a},\gamma^{tr}_{r,a}\})\geq 0$. Then $\triangle\lambda\in C^{tr}_a$ is the unique solution to $q^{tr}_a(\triangle\lambda)=0$ and $\sigma_1=\sigma_2=\sigma_3$ can be computed from (\ref{flow4123}). Moreover, $\triangle\lambda$ and $\sigma_1$  depend on $\sigma^{tr}_1$, $\sigma^{tr}_2$ and $\sigma^{tr}_3$ only through $\sigma^{tr}_1+\sigma^{tr}_2+\sigma^{tr}_3$.
\end{enumerate} 
The component $\bar{\varepsilon}^p$ can be computed from (\ref{correction_problem2}) for all of these cases.
\label{th_main}
\end{theorem}

\begin{theorem}
The discretized constitutive problem (\ref{k_step_problem}) has a unique solution. 
\label{th_main2}
\end{theorem}

\begin{remark}
\emph{Notice that Theorem \ref{th_main} contains the solution scheme to problem (\ref{correction_problem2}) and summarizes the advantages of the subdifferential treatment within the constitutive solution:
\begin{enumerate}
\item {\it Existence and uniqueness of the solution.} This expected result is  not usually discussed in literature.
\item {\it A priori known decision criteria.} Such criteria were known only for linear function $H$ (see, e.g., \cite{LR96}) where the solution components can be found in closed forms.
\item {\it Dependence of $\sigma_1,\sigma_2,\sigma_3, \triangle\lambda$ on $\sigma_1^{tr},\sigma_2^{tr},\sigma_3^{tr}$} has been described in more detail than it is known from literature. This enables us to simplify construction of the stress-strain and consistent tangent operators introduced in the next section, and discuss semismoothness of the stress-strain operator.
\end{enumerate}}
\label{remark_advantages}
\end{remark}

%%%%%%%%%%%%%%%%%%%%%%%%%%%%%%%%%%%%%%%%%%%%%%%%%%%%%%%%%%%%%%%%%5
\section{Stress-strain and consistent tangent operators}
\label{sec.Stress-strain_relation}

In this section, we extend the solution scheme from Theorem \ref{th_main} to problem (\ref{k_step_problem}) and define the stress-strain operator and its derivative, i.e., the {\it consistent tangent operator}. Beside the results from Section \ref{sec_time_discret}, we also use the framework from Section \ref{subsec_eigenprojection} based on eigenprojections and their derivatives.

The stress-strain relation can be represented by an implicit function $\mbf T$:
\begin{equation*}
\mbf{\sigma}(t_k):=\mbf T\left(\mbf{\varepsilon}(t_k);\mbf{\varepsilon^p}(t_{k-1}),\bar{\varepsilon}^p(t_{k-1})\right).
\label{T0}
\end{equation*}
If we fix step $k$ and recall $\bar{\varepsilon}^{p,tr}=\bar{\varepsilon}^p(t_{k-1})$, $\mbf\varepsilon^{tr}=\mbf\varepsilon(t_{k})-\mbf{\varepsilon^p}(t_{k-1})$, one can write
\begin{equation}
\mbf{\sigma}:=\mbf T\left(\mbf{\varepsilon};\mbf{\varepsilon^p}(t_{k-1}),\bar{\varepsilon}^p(t_{k-1})\right)=\mbf S\left(\mbf{\varepsilon}^{tr},\bar{\varepsilon}^{p,tr}\right)
\label{T}
\end{equation}
omitting the subscript $k$. The consistent tangent operator for step $k$ will be represented the Fr\' echet derivative $\mathcal D\mbf S\equiv\mathcal D_{\mbf\varepsilon^{tr}}\mbf S$. If it exists at $\left(\mbf{\varepsilon}^{tr},\bar{\varepsilon}^{p,tr}\right)$ then $\mathcal D_{\mbf{\varepsilon}}\mbf T=\mathcal D_{\mbf\varepsilon^{tr}} \mbf S$.
It is sufficient to derive the operators $\mbf S$ and $\mathcal D\mbf S$ on the following open sets:
\begin{eqnarray*}
M^{tr}_{e}&=&\{\mbf\varepsilon^{tr}\in\mathbb R^{3\times 3}_{sym}\; |\; q^{tr}(0)=q^{tr}_s(0)=f\left(\mbf{\sigma}^{tr},H(\bar{\varepsilon}^{p,tr})\right)<0\},\\
M^{tr}_{s}&=&\{\mbf\varepsilon^{tr}\in\mathbb R^{3\times 3}_{sym}\; |\; q^{tr}_s(0)>0,\; q^{tr}_s(\min\{\gamma^{tr}_{s,l},\gamma^{tr}_{s,r}\})< 0\},\\
M^{tr}_{l}&=&\{\mbf\varepsilon^{tr}\in\mathbb R^{3\times 3}_{sym}\; |\; \gamma^{tr}_{s,l}<\gamma^{tr}_{l,a},\; q^{tr}_l(\gamma^{tr}_{s,l})> 0,\; q^{tr}_l(\gamma^{tr}_{l,a})< 0\},\\
M^{tr}_{r}&=&\{\mbf\varepsilon^{tr}\in\mathbb R^{3\times 3}_{sym}\; |\; \gamma^{tr}_{s,r}<\gamma^{tr}_{r,a},\; q^{tr}_r(\gamma^{tr}_{s,r})> 0,\; q^{tr}_r(\gamma^{tr}_{r,a})< 0\},\\
M^{tr}_{a}&=&\{\mbf\varepsilon^{tr}\in\mathbb R^{3\times 3}_{sym}\; |\; q^{tr}_a(\max\{\gamma^{tr}_{l,a},\gamma^{tr}_{r,a}\})>0\}.
\end{eqnarray*}
From Section \ref{sec_time_discret}, it follows that these sets are mutually disjoint and the closure of their union is equal to $\mathbb R^{3\times 3}_{sym}$ since $\mbf{\sigma}^{tr}=\mathbb D_e:\mbf\varepsilon^{tr}$. Further, the tensors $\mbf{\sigma}^{tr}$ and $\mbf{\varepsilon}^{tr}$ have the same eigenvectors and their eigenvalues are related as follows:
\begin{equation}
\sigma_i^{tr}=\frac{1}{3}(3K-2G)(\varepsilon^{tr}_1+\varepsilon^{tr}_2+\varepsilon^{tr}_3)+2G\varepsilon^{tr}_i,\quad i=1,2,3.
\label{sigma_i^tr}
\end{equation}
Hence, $\varepsilon^{tr}_i>\varepsilon^{tr}_j$ if and only if $\sigma^{tr}_i>\sigma^{tr}_j$ for any $i,j=1,2,3$. Therefore, $\mbf{\sigma}^{tr}$ and $\mbf{\varepsilon}^{tr}$ also have the same eigenprojections.
For the sake of simplicity, we assume that $H$ is differentiable at $\bar{\varepsilon}^{p,tr}_{k}+\triangle\lambda (2\cos\phi)$ and denote $H_1:=H'(\bar{\varepsilon}^{p,tr}_{k}+\triangle\lambda (2\cos\phi))$.

\bigskip\noindent
{\it The elastic response.}  Let $\mbf{\varepsilon}^{tr}\in M^{tr}_{e}$. Then, clearly,
\begin{equation}
\mbf S\left(\mbf{\varepsilon}^{tr},\bar{\varepsilon}^{p,tr}\right)=\mathbb D_e:\mbf\varepsilon^{tr},\quad \mathcal D\mbf S\left(\mbf{\varepsilon}^{tr},\bar{\varepsilon}^{p,tr}\right)=\mathbb D_e.
\end{equation}

\bigskip\noindent
{\it The return to the smooth portion.} Let $\mbf{\varepsilon}^{tr}\in M^{tr}_{s}$. Then
$\varepsilon_1^{tr}>\varepsilon_2^{tr}>\varepsilon_3^{tr}$ holds and consequently, the values $\mbf E_i^{tr}:=\mbf E_i(\mbf\varepsilon^{tr})$, $\mathbb E_i^{tr}:=\mathbb E_i(\mbf\varepsilon^{tr})$, $i=1,2,3,$ are well-defined as follows from Section \ref{subsec_eigenprojection}. Therefore,
\begin{equation}
\mbf S\left(\mbf{\varepsilon}^{tr},\bar{\varepsilon}^{p,tr}\right)=\sum_{i=1}^3\sigma_i\mbf E_i^{tr},\quad \mathcal D\mbf S\left(\mbf{\varepsilon}^{tr},\bar{\varepsilon}^{p,tr}\right)=\sum_{i=1}^3\left[\sigma_i\mathbb E_i^{tr}+\mbf E_i^{tr}\otimes\mathcal D\sigma_i\right].
\end{equation}
Since,
\begin{eqnarray*}
\mathcal D\sigma_1&\stackrel{(\ref{flow11})}{=}&\frac{1}{3}(3K-2G)\mbf I+2G\mbf E_1^{tr}-\mathcal D(\triangle\lambda)\left[\frac{2}{3}(3K-2G)\sin\psi+2G(1+\sin\psi)\right],\\
\mathcal D\sigma_2&\stackrel{(\ref{flow12})}{=}&\frac{1}{3}(3K-2G)\mbf I+2G\mbf E_2^{tr}-\mathcal D(\triangle\lambda)\left[\frac{2}{3}(3K-2G)\sin\psi\right],\\
\mathcal D\sigma_3&\stackrel{(\ref{flow13})}{=}&\frac{1}{3}(3K-2G)\mbf I+2G\mbf E_3^{tr}-\mathcal D(\triangle\lambda)\left[\frac{2}{3}(3K-2G)\sin\psi-2G(1-\sin\psi)\right],
\end{eqnarray*}
we have
\begin{eqnarray}
\mathcal D\mbf S\left(\mbf{\varepsilon}^{tr},\bar{\varepsilon}^{p,tr}\right)&=&\sum_{i=1}^3\left[\sigma_i\mathbb E_i^{tr}+2G\mbf E_i^{tr}\otimes\mbf E_i^{tr}\right]+\frac{1}{3}(3K-2G)\mbf I\otimes\mbf I-\nonumber\\
&&-\left[2G(1+\sin\psi)\mbf E_1^{tr}-2G(1-\sin\psi)\mbf E_3^{tr}+\frac{2}{3}(3K-2G)\sin\psi\mbf I\right]\otimes\mathcal D(\triangle\lambda),\qquad
\end{eqnarray}
where 
\begin{equation*}
\mathcal D(\triangle\lambda)\stackrel{(\ref{q1})}{=}\frac{2G(1+\sin\phi)\mbf E_1^{tr}-2G(1-\sin\phi)\mbf E_3^{tr}+\frac{2}{3}(3K-2G)\sin\phi\mbf I}{\frac{4}{3}(3K-2G)\sin\psi\sin\phi+4G(1+\sin\psi\sin\phi)+4H_1\cos^2\phi}.
\end{equation*}

\bigskip\noindent
{\it The return to the left edge.} Let $\mbf{\varepsilon}^{tr}\in M^{tr}_{l}$. Then $\sigma_1=\sigma_2$ and $\varepsilon_2^{tr}>\varepsilon_3^{tr}$. It means that one can introduce the notation $\mbf E_3^{tr}:=\mbf E_3(\mbf\varepsilon^{tr})$, $\mbf E_{12}^{tr}:=\mbf I-\mbf E_3^{tr}$, $\mathbb E_3^{tr}:=\mathbb E_3(\mbf\varepsilon^{tr})$ and write 
\begin{equation}
\mbf S\left(\mbf{\varepsilon}^{tr},\bar{\varepsilon}^{p,tr}\right)=\sigma_1\mbf E_{12}^{tr}+\sigma_3\mbf E_3^{tr},\quad \mathcal D\mbf S\left(\mbf{\varepsilon}^{tr},\bar{\varepsilon}^{p,tr}\right)=(\sigma_3-\sigma_1)\mathbb E_3^{tr}+\mbf E_{12}^{tr}\otimes\mathcal D\sigma_1+\mbf E_{3}^{tr}\otimes\mathcal D\sigma_3.
\label{sigmaE2}
\end{equation}
Since,
\begin{eqnarray*}
\mathcal D\sigma_1&\stackrel{(\ref{flow212})}{=}&\frac{1}{3}(3K-2G)\mbf I+G\mbf E_{12}^{tr}-\mathcal D(\triangle\lambda)\left[\frac{2}{3}(3K-2G)\sin\psi+G(1+\sin\psi)\right],\\
\mathcal D\sigma_3&\stackrel{(\ref{flow23})}{=}&\frac{1}{3}(3K-2G)\mbf I+2G\mbf E_3^{tr}-\mathcal D(\triangle\lambda)\left[\frac{2}{3}(3K-2G)\sin\psi-2G(1-\sin\psi)\right],
\end{eqnarray*}
we have
\begin{eqnarray}
\mathcal D\mbf S\left(\mbf{\varepsilon}^{tr},\bar{\varepsilon}^{p,tr}\right)&=&(\sigma_3-\sigma_1)\mathbb E_3^{tr}+G\mbf E_{12}^{tr}\otimes\mbf E_{12}^{tr}+2G\mbf E_3^{tr}\otimes\mbf E_3^{tr}+\frac{1}{3}(3K-2G)\mbf I\otimes\mbf I-\nonumber\\
&&-\left[G(1+\sin\psi)\mbf E_{12}^{tr}-2G(1-\sin\psi)\mbf E_3^{tr}+\frac{2}{3}(3K-2G)\sin\psi\mbf I\right]\otimes\mathcal D(\triangle\lambda),\qquad
\end{eqnarray}
where 
\begin{equation*}
\mathcal D(\triangle\lambda)\stackrel{(\ref{q2})}{=}\frac{G(1+\sin\phi)\mbf E_{12}^{tr}-2G(1-\sin\phi)\mbf E_3^{tr}+\frac{2}{3}(3K-2G)\sin\phi\mbf I}{\frac{4}{3}(3K-2G)\sin\psi\sin\phi+G(1+\sin\psi)(1+\sin\phi)+2G(1-\sin\psi)(1-\sin\phi)+4H_1\cos^2\phi}.
\end{equation*}

\bigskip\noindent
{\it The return to the right edge.} Let $\mbf{\varepsilon}^{tr}\in M^{tr}_{r}$.   Then  $\sigma_2=\sigma_3$ and $\varepsilon_1^{tr}>
\varepsilon_2^{tr}$. It means that one can introduce the notation $\mbf E_1^{tr}:=\mbf E_1(\mbf\varepsilon^{tr})$, $\mbf E_{23}^{tr}:=\mbf I-\mbf E_1^{tr}$, $\mathbb E_1^{tr}:=\mathbb E_1(\mbf\varepsilon^{tr})$ and write 
\begin{equation}
\mbf S\left(\mbf{\varepsilon}^{tr},\bar{\varepsilon}^{p,tr}\right)=\sigma_1\mbf E_{1}^{tr}+\sigma_3\mbf E_{23}^{tr},\quad \mathcal D\mbf S\left(\mbf{\varepsilon}^{tr},\bar{\varepsilon}^{p,tr}\right)=(\sigma_1-\sigma_3)\mathbb E_1^{tr}+\mbf E_{1}^{tr}\otimes\mathcal D\sigma_1+\mbf E_{23}^{tr}\otimes\mathcal D\sigma_3.
\label{sigmaE3}
\end{equation}
Since,
\begin{eqnarray*}
\mathcal D\sigma_1&\stackrel{(\ref{flow31})}{=}&\frac{1}{3}(3K-2G)\mbf I+2G\mbf E_{1}^{tr}-\mathcal D(\triangle\lambda)\left[\frac{2}{3}(3K-2G)\sin\psi+2G(1+\sin\psi)\right],\\
\mathcal D\sigma_3&\stackrel{(\ref{flow323})}{=}&\frac{1}{3}(3K-2G)\mbf I+G\mbf E_{23}^{tr}-\mathcal D(\triangle\lambda)\left[\frac{2}{3}(3K-2G)\sin\psi-G(1-\sin\psi)\right],
\end{eqnarray*}
we have
\begin{eqnarray}
\mathcal D\mbf S\left(\mbf{\varepsilon}^{tr},\bar{\varepsilon}^{p,tr}\right)&=&(\sigma_1-\sigma_3)\mathbb E_1^{tr}+2G\mbf E_{1}^{tr}\otimes\mbf E_{1}^{tr}+G\mbf E_{23}^{tr}\otimes\mbf E_{23}^{tr}+\frac{1}{3}(3K-2G)\mbf I\otimes\mbf I-\nonumber\\
&&-\left[2G(1+\sin\psi)\mbf E_{1}^{tr}-G(1-\sin\psi)\mbf E_{23}^{tr}+\frac{2}{3}(3K-2G)\sin\psi\mbf I\right]\otimes\mathcal D(\triangle\lambda),\qquad
\end{eqnarray}
where 
\begin{equation*}
\mathcal D(\triangle\lambda)\stackrel{(\ref{q3})}{=}\frac{2G(1+\sin\phi)\mbf E_{1}^{tr}-G(1-\sin\phi)\mbf E_{23}^{tr}+\frac{2}{3}(3K-2G)\sin\phi\mbf I}{\frac{4}{3}(3K-2G)\sin\psi\sin\phi+2G(1+\sin\psi)(1+\sin\phi)+G(1-\sin\psi)(1-\sin\phi)+4H_1\cos^2\phi}.
\end{equation*}

\bigskip\noindent
{\it The return to the apex.} Let $\mbf{\varepsilon}^{tr}\in M^{tr}_{a}$. Then $\sigma_1=\sigma_2=\sigma_3:=p$ and
\begin{equation}
\mbf S\left(\mbf{\varepsilon}^{tr},\bar{\varepsilon}^{p,tr}\right)=p\mbf I,\quad p=p^{tr}-(2K\sin\psi)\triangle\lambda,\quad p^{tr}=\frac{1}{3}(\sigma_1^{tr}+\sigma_2^{tr}+\sigma_3^{tr})=K(\varepsilon_1^{tr}+\varepsilon_2^{tr}+\varepsilon_3^{tr}),
\label{sigmaE4}
\end{equation}
\begin{equation*}
\mathcal D\mbf S\left(\mbf{\varepsilon}^{tr},\bar{\varepsilon}^{p,tr}\right)\stackrel{(\ref{sigmaE4})}{=}\frac{\partial p}{\partial p^{tr}}K\mbf I\otimes\mbf I=\left(1-2K\sin\psi\frac{\partial\triangle\lambda}{\partial p^{tr}}\right)K\mbf I\otimes\mbf I.
\end{equation*}
Here, we use $\frac{\partial p^{tr}}{\partial \mbf{\varepsilon}^{tr}}=K\mbf I$. From the implicit equation $q^{tr}_a(\triangle\lambda)=0$, we obtain
\begin{equation*}
\frac{\partial\triangle\lambda}{\partial p^{tr}}\stackrel{(\ref{q4})}{=}\frac{\sin\phi}{2K\sin\psi\sin\phi+2H_1\cos^2\phi}.
\label{dlambda4}
\end{equation*}
Hence,
\begin{equation}
\mathcal D\mbf S\left(\mbf{\varepsilon}^{tr},\bar{\varepsilon}^{p,tr}\right)=K\left(1-\frac{K\sin\psi\sin\phi}{K\sin\psi\sin\phi+H_1\cos^2\phi}\right)\mbf I\otimes\mbf I.
\end{equation}

\begin{remark}
\emph{For each of the return type, we derived just one formula for $\mathcal D\mbf S$ without any other branching that depends on multiplicity of $\varepsilon_1^{tr},\varepsilon_2^{tr},\varepsilon_3^{tr}$. This was achieved due to deeper analysis of dependencies within the constitutive solution, see Theorem \ref{th_main}. The additional branching in $\mathcal D\mbf S$ is introduced, e.g., in \cite[Appendix A]{NPO08}. In many other references, $\mathcal D\mbf S$ is correctly derived only under the assumption $\varepsilon_1^{tr}>\varepsilon_2^{tr}>\varepsilon_3^{tr}$. However, such formulas can cause significant rounding errors in vicinity of the multiple eigenvalues.}
\end{remark}

\begin{remark}
\emph{Notice that one can continuously extend the definition of $\mbf T\left(\cdot\,;\mbf{\varepsilon^p}(t_{k-1}),\bar{\varepsilon}^p(t_{k-1})\right)=\mbf S(\cdot\,,\bar\varepsilon^{p,tr})$ on $\mathbb R^{3\times 3}_{sym}\setminus (M^{tr}_{e}\cup M^{tr}_{s}\cup M^{tr}_{l}\cup M^{tr}_{r}\cup M^{tr}_{a}).$  Further, one can investigate semismoothness of $\mbf T\left(\cdot\,;\mbf{\varepsilon^p}(t_{k-1}),\bar{\varepsilon}^p(t_{k-1})\right)$  in $\mathbb R^{3\times 3}_{sym}$. This property ensures superlinear convergence of algorithms introduced in the next section. To show the semismoothness in $M^{tr}_{e}$, $M^{tr}_{s}$, $M^{tr}_{l}$, $M^{tr}_{r}$, and $M^{tr}_{a}$, one can use a standard framework introduced, e.g. in \cite{ GrVa09, SaWi11, Sy09,CKSV14, Sy14,SCKKZB15}.  At the remaining points, the semismoothness is also expected based on Remarks \ref{rem_semismooth} and \ref{remark_advantages}$_3$. However, its eventual proof seems to be more involved and we will skip it for the sake of brevity.} 
\end{remark}

Below, we use the notation $\mathbb T\left(\cdot\,;\mbf{\varepsilon^p}(t_{k-1}),\bar{\varepsilon}^p(t_{k-1})\right)$ for the Clark generalized derivative of $\mbf T$ with respect to the strain tensor. Clearly, $\mathbb T\left(\mbf{\varepsilon};\mbf{\varepsilon^p}(t_{k-1}),\bar{\varepsilon}^p(t_{k-1})\right)=\mathcal D_{\mbf\varepsilon}\mbf T\left(\mbf{\varepsilon};\mbf{\varepsilon^p}(t_{k-1}),\bar{\varepsilon}^p(t_{k-1})\right)$ when $\mbf T\left(\cdot\,;\mbf{\varepsilon^p}(t_{k-1}),\bar{\varepsilon}^p(t_{k-1})\right)$ is differentiable at $\mbf{\varepsilon}$.

%********************** Section 6**********************************************************************

\section{Direct and indirect methods of incremental limit analysis}
\label{sec_realization}

Inserting the stress-strain operator $\mbf T$ to the balance equation, we obtain the incremental boundary value elastoplastic problem \cite{NPO08,SCKKZB15}. This problem is further discretized in space by the finite element method and combined with the limit load analysis as it is usual in Mohr-Coulomb plasticity \cite{NPO08, CL90} . In this section, we introduce the direct and indirect methods of the incremental limit analysis. For the sake of brevity, we focus only on an algebraic formulation of the problem. 

The vector of  internal forces and the consistent tangent stiffness matrix at the $k$-th step are represented by functions $\mbf F_k:\mathbb R^n\rightarrow\mathbb R^n$ and $\mbf K_k:\mathbb R^n\rightarrow\mathbb R^{n\times n}$, respectively. It is worth mentioning that $\mbf F_k$ and $\mbf K_k$ are assembled using the operators $\mbf T$ and $\mathbb T$ at each integration point \cite{SCKKZB15}. Notice that the algebraic representation of the used second and fourth order tensors is introduced in Appendix B. Further, we consider the load of external forces at step $k$ in the form $\zeta_k\mbf l$ where $\mbf l\in\mathbb R^n$ is fixed and $\zeta_k:=\zeta(t_k)$. Then the $k$-step problem reads as:
$$(\mathcal P_k)_\zeta \qquad\mbox{given }\zeta_{k}\in\mathbb R_+, \mbox{ find } \mbf u_{k}\in\mathbb R^n:\quad \mbf{F}_{k}(\mbf{u}_{k})=\zeta_{k}\mbf{l},$$ 
where $\mbf u_k$ is the displacement vector.
We assume that the parameters $\zeta$ and $t$ coincide and their limit value $\zeta_{lim}$ is unknown. It is well known that the investigated body collapses when this critical (limit) value is exceeded. Therefore, $\zeta_{lim}$ is an important safety parameter and beyond $\zeta_{lim}$ no solution exists. Possibly $\zeta_{lim}=+\infty$, however in meaningful settings of the problem, $\zeta_{lim}$ is finite. The simplest computational technique is based on the so-called {\it incremental limit analysis} where we adaptively construct the sequence
$$0<\zeta_1<\zeta_2<\ldots<\zeta_{k}<\zeta_{k+1}<\ldots<\zeta_{lim}$$ 
depending on solvability of $(\mathcal P_k)_\zeta$ to detect inadmissible load factors. In practice, the increment of $\zeta_k$ decreases when a chosen numerical method does not converge at step $k$. Such blind determination of $\zeta_k$ is an evident drawback of this {\it direct method}.

More sophisticated adaptive strategy  is based on local and/or global material response of the body on the prescribed load history. To this end, we compute the values $\alpha_k=\mbf b^T\mbf u_k$, $k=1,2,\ldots$ where $\mbf u_k$ is the solution to $(\mathcal P_k)_\zeta$ and $\mbf b$ is chosen so that to be the sequence $\{\alpha_k\}$ increasing. There are many ways how to do it. For example, one can detect a point on the investigated body where it is expected that a selected displacement is the most sensitive on the applied forces. Then $\mbf b$ is the restriction of the displacement vector to its component. More universally, one can also set $\mbf b=\mbf l$. This choice represents the work of external forces, is meaningful even for continuous setting of the problem and was analyzed in \cite{SHHC15, CHKS15, HRS16, HRS16b} for generalized Hencky's plasticity. Clearly, if the increment $\alpha_k-\alpha_{k-1}$ significantly enlarges with increasing $k$ then it is convenient to reduce the increment of $\zeta$ for the next step.

The knowledge of suitable $\mbf b$ also enables to introduce the {\it indirect method} of incremental limit analysis where the increasing sequence $\{\alpha_k\}$ is given and the sequences $\{\zeta_k\}$ and $\{\mbf u_k\}$ are computed using the following auxilliary problem:
$$(\mathcal P_k)^\alpha \qquad \mbox{given }(\mbf b,\alpha_{k})\in\mathbb R^n\times\mathbb R_+,\;\mbox{find } (\mbf u_{k}, \zeta_{k})\in\mathbb R^n\times\mathbb R_+:\quad \left\{
\begin{array}{c}
\mbf{F}_{k}(\mbf{u}_{k})=\zeta_{k}\mbf{l},\\[1mm]
\mbf b^T\mbf u_{k}=\alpha_{k}.
\end{array}\right.$$ 
Clearly, if $(\mbf u_{k}, \zeta_{k})$ is the solution to $(\mathcal P_k)^\alpha$ then $\mbf u_{k}$ also solves $(\mathcal P_k)_\zeta$ for $\zeta_k$ and $\zeta_k\leq\zeta_{lim}\leq+\infty$.
Unlike to problem $(\mathcal P_k)_\zeta$, one can expect that problem $(\mathcal P_k)^\alpha$ has the solution for any $\alpha_k$. Since the parameter $\alpha$ can be enlarged arbitrary, the indirect method is more stable and does not include any blind guesswork unlike the direct one.
This is the main advantage of the indirect method. For the associative Mohr-Coulomb model, one can expect that $\zeta_k\rightarrow\zeta_{lim}$ as $\alpha_k\rightarrow+\infty$. This is proven in \cite{CHKS15, HRS16} for $\mbf b=\mbf l$ and the generalized Hencky's plasticity. For the nonassociative Mohr-Coulomb model with $\psi<<\phi$, we observe that $\zeta_{\tilde k}\approx\zeta_{lim}$ for some finite $\tilde k$ and for $k>\tilde k$, the sequence $\{\zeta_k\}$ is nonincreasing. In such a case, the material exhibits softening behavior and the direct method is too convenient. It is also worth mentioning that the indirect method is similar to the arc-length method introduced, e.g., in \cite{C97, NPO08}.

We solve problems $(\mathcal P)_\zeta$ and $(\mathcal P)^\alpha$ by the semismooth Newton method:

\begin{algorithm}[ALG-$\zeta$]
\hspace{0.2cm}
\begin{spacing}{1.2}
\begin{algorithmic}[1]
  \STATE initialization: $\mbf{u}_{k}^0$
  \FOR{$i=0,1,2,\ldots$}
    \STATE find $\delta \mbf{u}^{i}\in\mbf{V}$: $\;\mbf {K}_k(\mbf{u}_{k}^i)\delta  \mbf{u}^{i}=\zeta_k\mbf{l}-\mbf{F}_{k}(\mbf{u}_{k}^i)$
    \STATE compute $\mbf{u}_{k}^{i+1}=\mbf{u}_{k}^i+\delta \mbf{u}^{i}$
    \STATE {\bf if
    }{$\|\delta
    \mbf{u}^{i}\|/(\|\mbf{u}_{k}^{i+1}\|+\|\mbf{u}_{k}^i\|)\leq\epsilon_{Newton}$} {\bf then stop}
  \ENDFOR
  \STATE set $\mbf{u}_{k}=\mbf{u}_{k}^{i+1}$.
\end{algorithmic}
\end{spacing}
\end{algorithm}

\begin{algorithm}[ALG-$\alpha$]
\hspace{0.2cm}
\begin{spacing}{1.2}
\begin{algorithmic}[1]
  \STATE initialization: $\mbf{u}_{k}^0$, $\zeta_k^0$
  \FOR{$i=0,1,2,\ldots$}
    \STATE find $\mbf{v}^{i},\;\mbf{w}^{i}\in\mbf{V}$: $\;\mbf {K}_k(\mbf{u}_{k}^i)\mbf{v}^{i}=\zeta_k^i\mbf{l}-\mbf{F}_{k}(\mbf{u}_{k}^i),\;\;\mbf {K}_k(\mbf{u}_{k}^i)\mbf{w}^{i}=\mbf{l}$
    \STATE compute $\delta\zeta^i=[\alpha_k-\mbf b^T(\mbf{u}_{k}^i+\mbf{v}^{i})]/\mbf b^T\mbf{w}^{i}$
    \STATE compute $\delta \mbf{u}^{i}=\mbf{v}^{i}+\delta\zeta^i\mbf{w}^{i}$
    \STATE set $\mbf{u}_{k}^{i+1}=\mbf{u}_{k}^i+\delta \mbf{u}^{i},\; \zeta_{k}^{i+1}=\zeta_{k}^i+\delta \zeta^{i}$
    \STATE {\bf if
    }{$\|\delta
    \mbf{u}^{i}\|/(\|\mbf{u}_{k}^{i+1}\|+\|\mbf{u}_{k}^i\|)\leq\epsilon_{Newton}$} {\bf then stop}
  \ENDFOR
  \STATE set $\mbf{u}_{k}=\mbf{u}_{k}^{i+1}$, $\zeta_k=\zeta_k^{i+1}$.
\end{algorithmic}
\end{spacing}
\end{algorithm}
If $\mbf T\left(\cdot\,;\mbf{\varepsilon^p}(t_{k-1}),\bar{\varepsilon}^p(t_{k-1})\right)$  is semismoothness in $\mathbb R^{3\times 3}_{sym}$ then one can easily show that $\mbf F_k$ is semismooth in $\mathbb R^n$. The semismoothness is an essential assumption ensuring local superlinear convergence of these algorithms (see, e.g., \cite{CHKS15}). Further, we initialize ALG-$\zeta$ and ALG-$\alpha$ using the linear extrapolation of the solutions from two previous steps. In particular, we prescribe
$$\mbf{u}_{k}^0=\mbf u_{k-1}+\frac{\alpha_k-\alpha_{k-1}}{\alpha_{k-1}-\alpha_{k-2}}(\mbf u_{k-1}-\mbf u_{k-2}),\quad \zeta_{k}^0=\zeta_{k-1}+\frac{\alpha_k-\alpha_{k-1}}{\alpha_{k-1}-\alpha_{k-2}}(\zeta_{k-1}-\zeta_{k-2})$$
in ALG-$\alpha$  for $k\geq 2$, and analogously, in ALG-$\zeta$. We observe that this initialization is more convenient than $\mbf{u}_{k}^0=\mbf u_{k-1}$, $\zeta_k^0=\zeta_{k-1}$.

The direct and indirect methods of incremental limit analysis are compared in Section \ref{subsec_comparison}.

%********************** Section 7**********************************************************************

\section{Numerical experiments - slope stability}
\label{sec_experiments}

We have implemented the direct and indirect methods of incremental limit analysis in MatLab for 3D slope stability problem and its plane strain reduction. These experimental codes denoted as SS-MC-NP-3D, SS-MC-NH and SS-MC-NH-Acontrol are available in \cite{Mcode}. The codes are vectorized and include the improved return-mapping scheme for the Mohr-Coulomb model in combination with ALG-$\zeta$ or ALG-$\alpha$.  One can choose: a) several types of finite elements with appropriate numerical quadratures; b) locally refined meshes with various densities. 

\begin{figure}[htbp]
\center
  \includegraphics[width=0.6\textwidth]{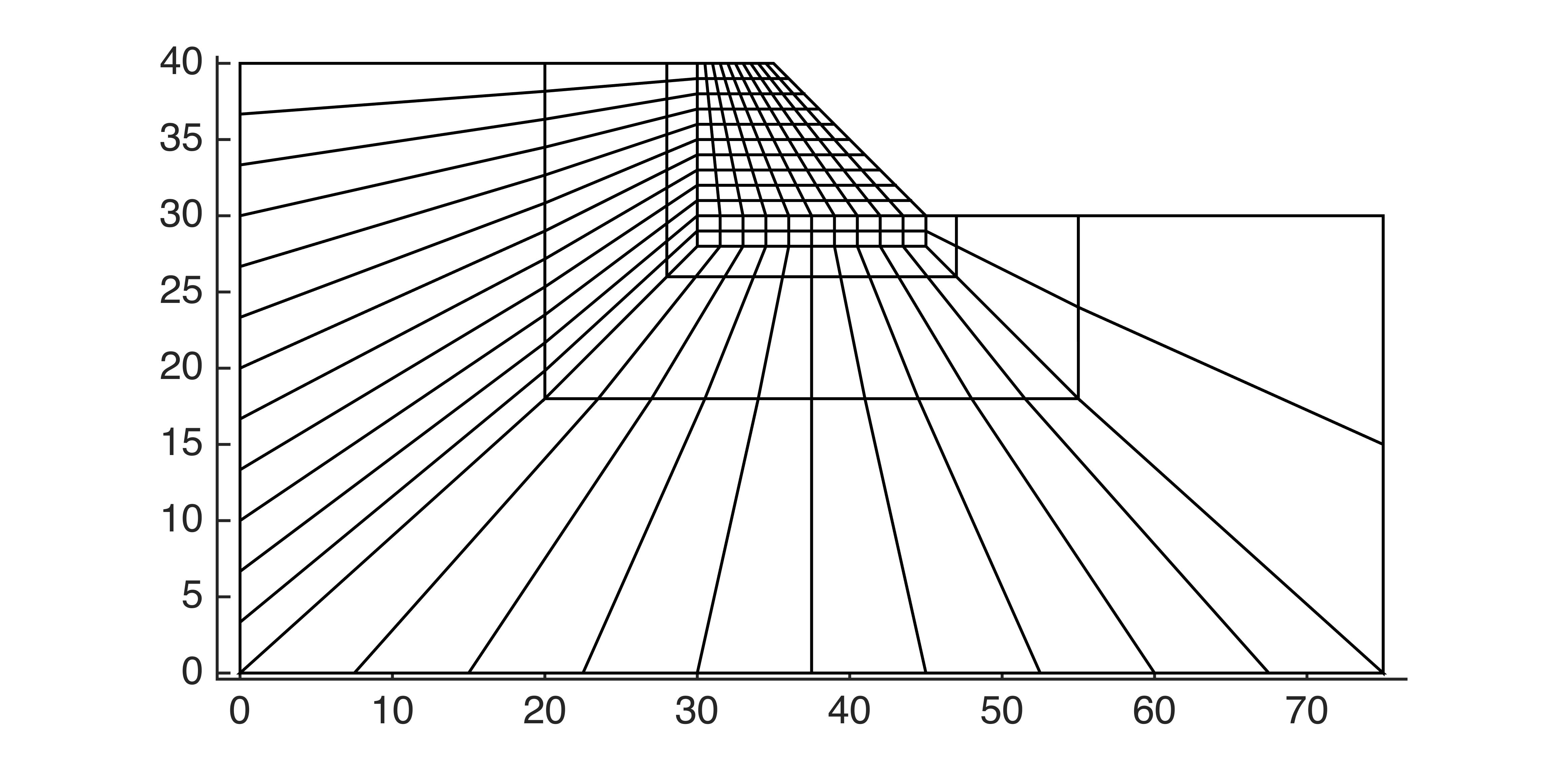}
   \caption{\small{Cross section of the body with the coarsest mesh for $Q2$ elements.}}
   \label{fig.mesh_Q2}
\end{figure}

We consider the benchmark plane strain problem introduced in \cite[Page 351]{NPO08} and its extension for 3D case.
The 2D cross-section of the body with the coarsest mesh considered in \cite[SS-MC-NH]{Mcode} is depicted in Figure \ref{fig.mesh_Q2}. The 3D geometry and the corresponding hexahedral mesh arise from 2D by extruding. The slope height is 10 m and its inclination is $45^\circ$. On the bottom, we assume that the body is fixed and, on the lateral sides,  zero normal displacements are prescribed. The body is subjected to self-weight. We set the specific weight $\rho g=20\, $kN/m$^3$ with $\rho$ being the mass density and $g$ the gravitational acceleration. Such a volume force is multiplied by the load factor $\zeta$. The parameter $\alpha$ is here the settlement at the corner point $A$ on the top of the slope to be in accordance with \cite{NPO08}.
Further, we set $E=20\,000\,$kPa, $\nu=0.49$, $\phi=20^{\circ}$ and $c=50\,$kPa, where $c$ denotes the cohesion for the perfect plastic model. Hence, $G=67\,114\, $kPa and $K=3\ 333\ 333\,$kPa. The remaining parameters of the Mohr-Coulomb model will be introduced below depending on a particular experiment.

We introduce one experiment for the plane strain (2D) problem and one for the 3D problem. The primary aim of these experiments is to numerically illustrate that the formulas derived in Sections \ref{sec_time_discret}, \ref{sec.Stress-strain_relation} and Appendix A work well. This can be confirmed by observing the superlinear convergence of ALG-$\zeta$ and ALG-$\alpha$ and their stability in vicinity of the limit load. We also prescribe a high precision of these algorithms by the setting $\epsilon_{Newton}=10^{-12}$ in both experiments. Other aims will be specified below.

\subsection{Comparison of the direct and indirect methods in 2D}
\label{subsec_comparison}

We compare the direct method (code SS-MC-NH) and the indirect method (code SS-MC-NH-Acontrol) of the incremental limit analysis on the slope stability benchmark in 2D. We consider the associative Mohr-Coulomb model containing the nonlinear isotropic hardening defined as in \cite{SCKKZB15}:
$$H(\bar\varepsilon^p)=\min\left\{c-c_0,\;\tilde H\bar\varepsilon^p-\frac{\tilde H^2}{4(c-c_0)}(\bar\varepsilon^p)^2\right\},\quad c_0=40\, \mbox{kPa}, \; \tilde H=10000\,\mbox{kPa}.$$
Here, $\tilde H$ represents the initial slope of $H$ and the material response is perfect plastic for sufficiently large values of $\bar\varepsilon^p$. The function $H$ is smooth and its influence on the limit load factor is negligible based on expertise introduced in \cite{SCKKZB15}. We set $\psi=\phi$ to have the associative model. 

Further, we use the $Q2$ elements (i.e. eight-noded quadrilaterals) with $3\times 3$ integration quadrature and the mesh with 37265 nodal points including the midpoints and with 110592 integration points. The mesh has a similar scheme as in  Figure \ref{fig.mesh_Q2} but, of course, it is much more finer. Since the Matlab code is vectorized, we fix 10 inner Newton's iterations for finding the unknown plastic multipliers in each integration point.

Recall that in each step of the direct method, we solve problem $(\mathcal P)_\zeta$ using ALG-$\zeta$. We set the initial load increment $\delta\zeta_0=0.5$.  If  ALG-$\zeta$ converges during 50 iterations for step $k\geq 1$ and if the computed increment of the settlement satisfies $\alpha_k-\alpha_{k-1}<0.5\,$m  then we set $\delta\zeta_{k+1}=\delta\zeta_k$. Otherwise, the increment is divided by two. Within the indirect method where  problem $(\mathcal P)^\alpha$ is solved using ALG-$\alpha$ we set the initial increment $\delta\alpha_0=0.0414$ of the settlement to have comparable results with the direct method. If the computed load increment satisfies $|\zeta_k-\zeta_{k-1}|>5e-3$ then we set $\delta\alpha_{k+1}=\delta\alpha_{k}$. Otherwise, $\delta\alpha_{k+1}=2\delta\alpha_{k}$. The loading process is terminated when the computed settlement exceeds 4 meters for both methods. 

\begin{figure}[htbp]
\begin{minipage}[t]{0.47\textwidth}
  \center
  \includegraphics[width=\textwidth]{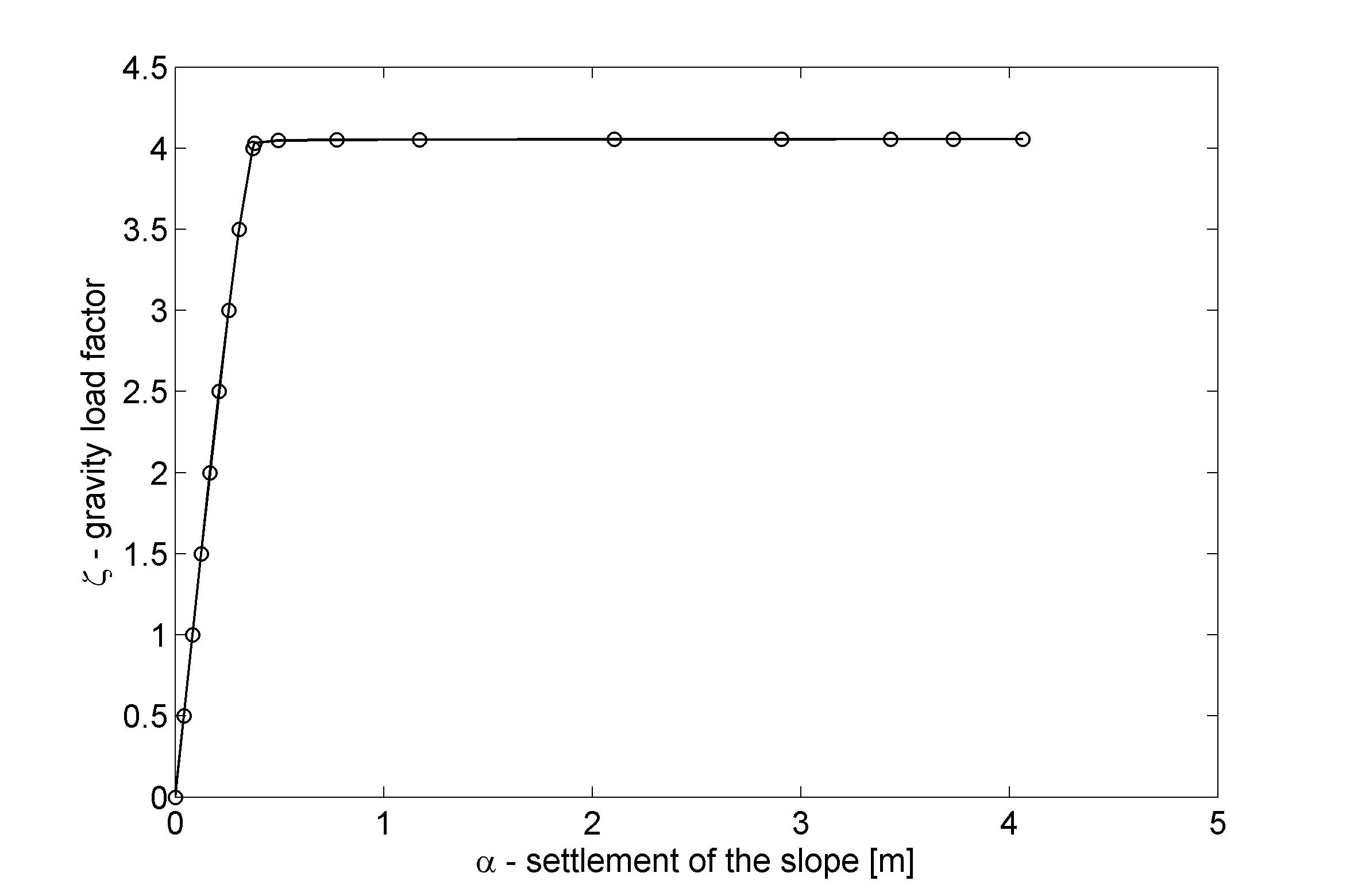}
   \caption{\small{Load path for the direct method.}}
   \label{fig.load_path}
\end{minipage}
\hfill
\begin{minipage}[t]{0.47\textwidth}
  \center
   \includegraphics[width=\textwidth]{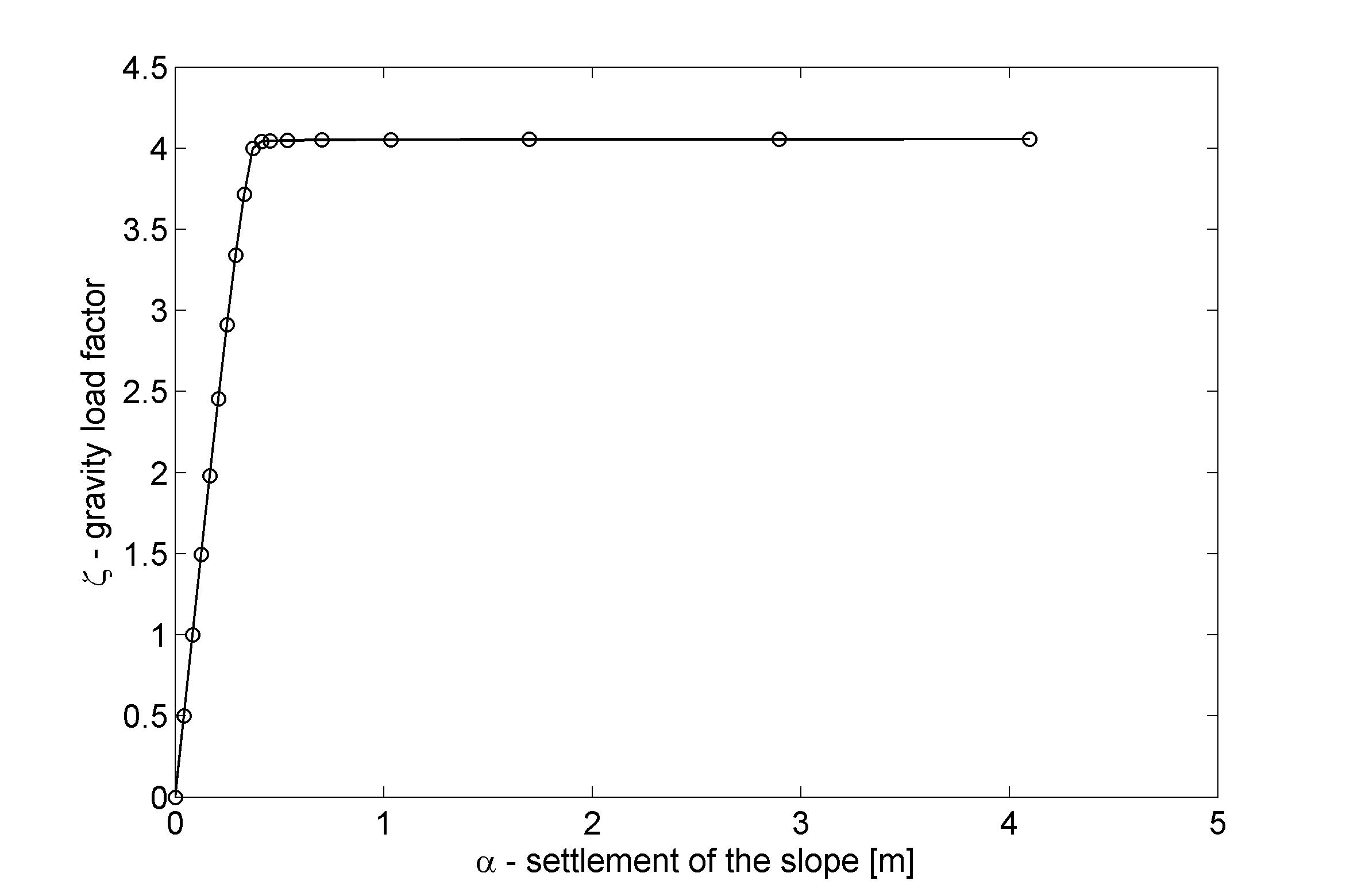}
   \caption{\small{Load path for the indirect method.}}
   \label{fig.load_path_A-control}
\end{minipage}
\end{figure}

\begin{figure}[htbp]
\begin{minipage}[t]{0.47\textwidth}
  \center
  \includegraphics[width=\textwidth]{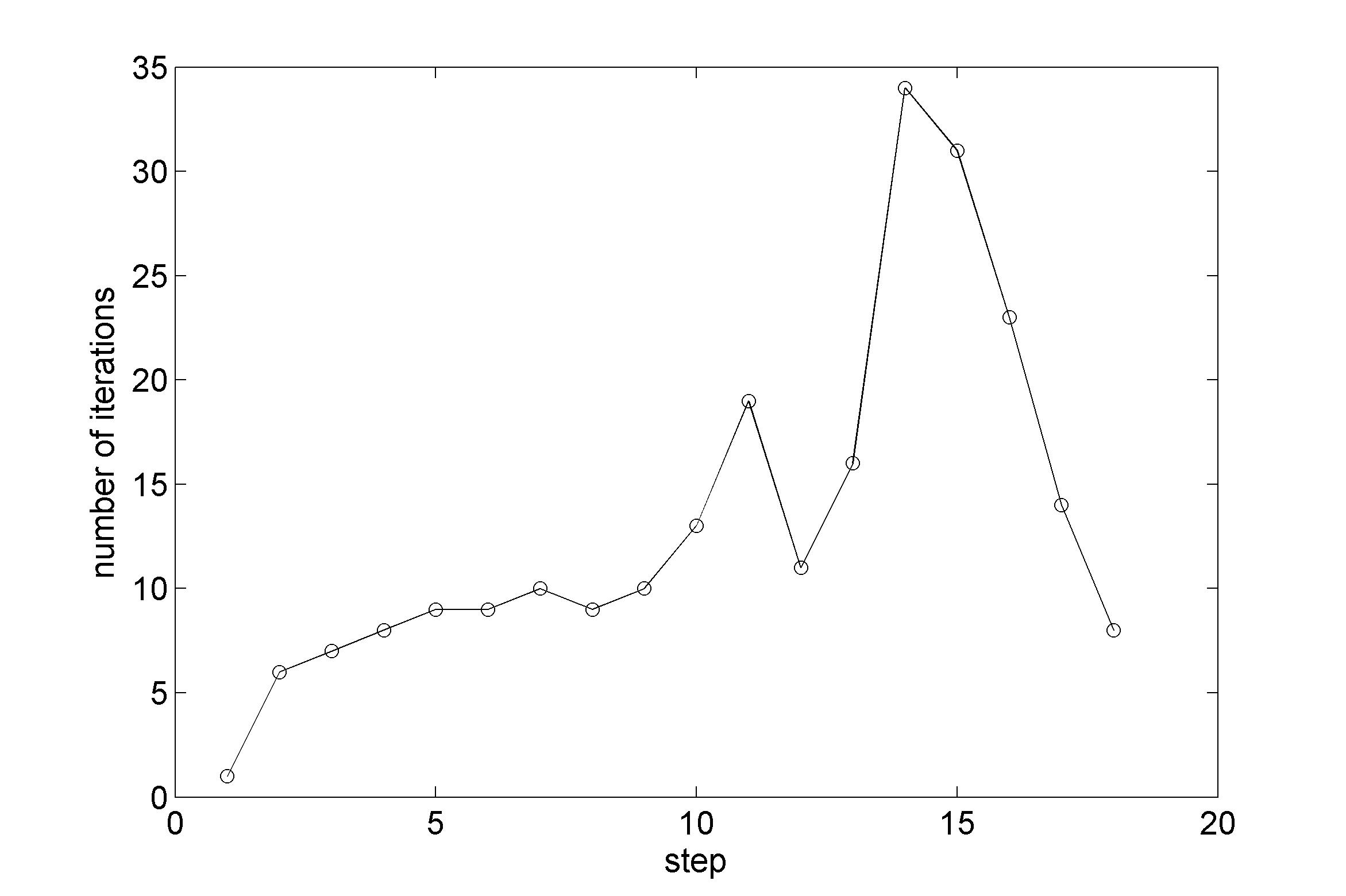}
   \caption{\small{Number of iterations for ALG-$\zeta$.}}
   \label{fig.iterations}
\end{minipage}
\hfill
\begin{minipage}[t]{0.47\textwidth}
  \center
   \includegraphics[width=\textwidth]{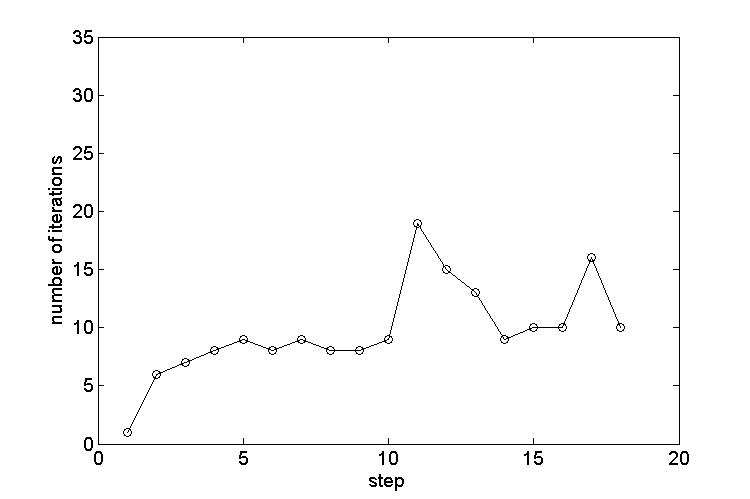}
   \caption{\small{Number of iterations for ALG-$\alpha$.}}
   \label{fig.iterations_A-control}
\end{minipage}
\end{figure}

\begin{figure}[htbp]
\begin{minipage}[t]{0.47\textwidth}
  \center
  \includegraphics[width=\textwidth]{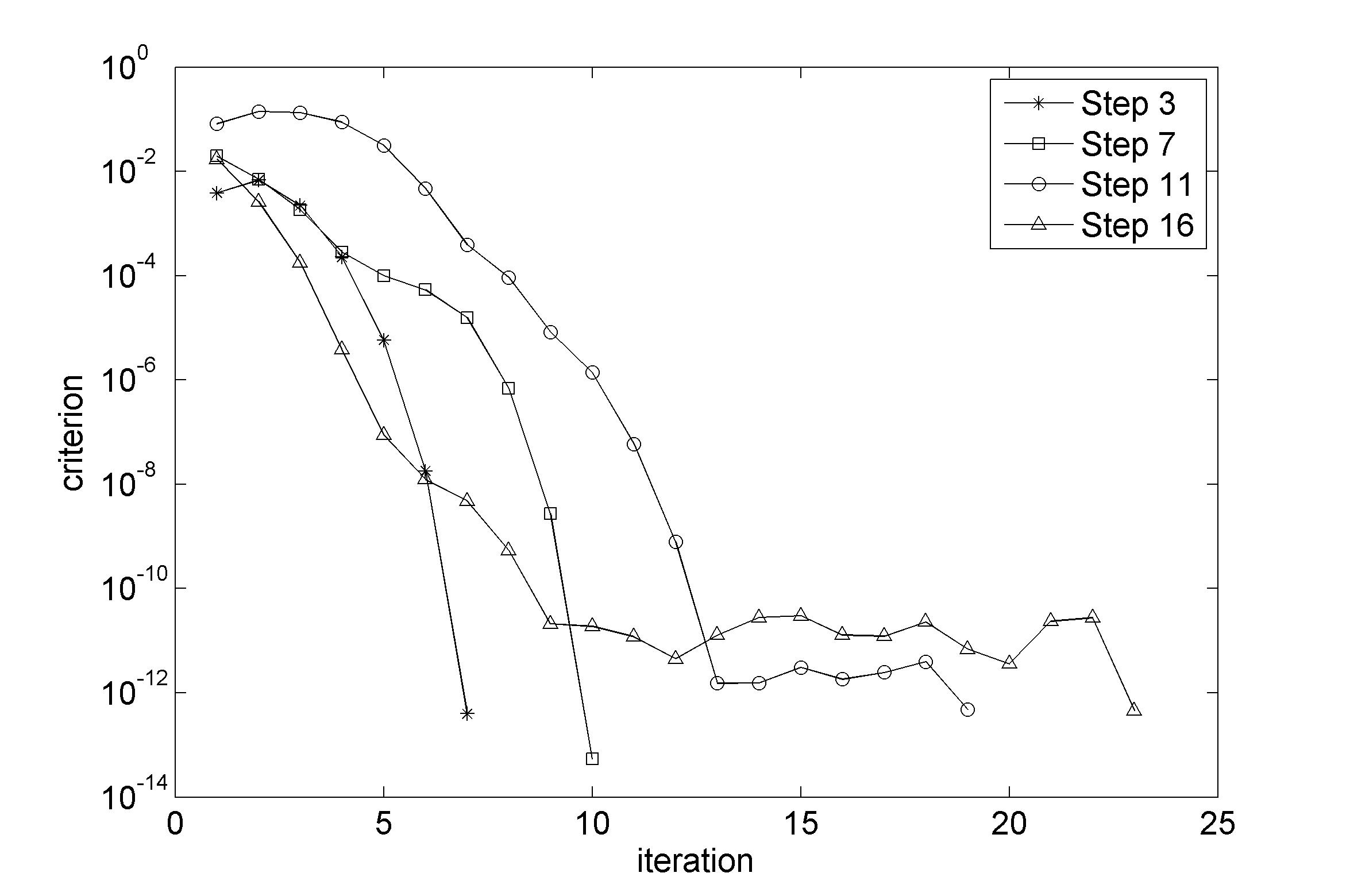}
   \caption{\small{Convergence at selected steps for the direct method.}}
   \label{fig.convergence}
\end{minipage}
\hfill
\begin{minipage}[t]{0.47\textwidth}
  \center
   \includegraphics[width=\textwidth]{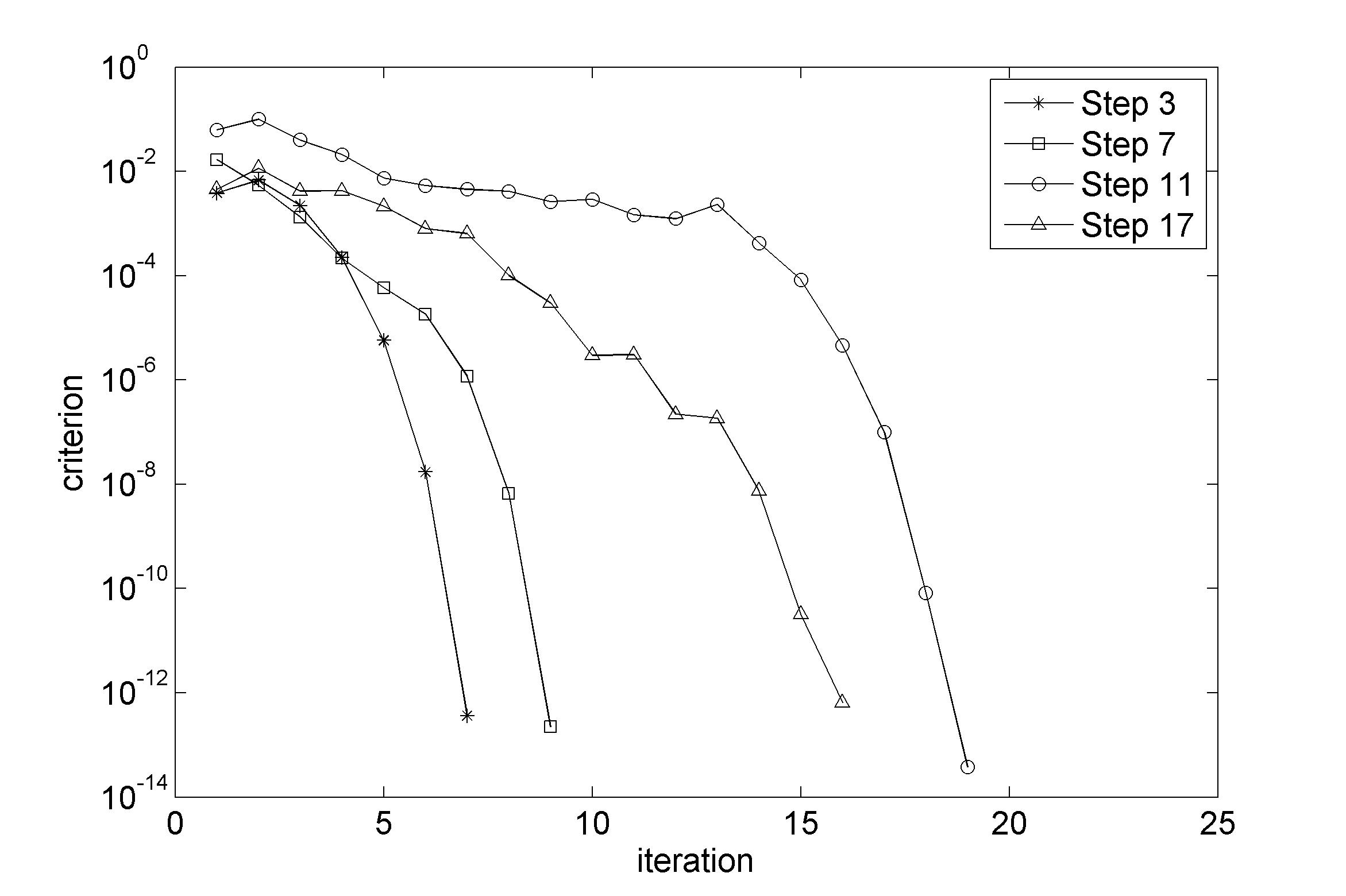}
   \caption{\small{Convergence at selected steps for the indirect method.}}
   \label{fig.convergence_A-control}
\end{minipage}
\end{figure}

The comparison of the direct and indirect methods is depicted in Figures \ref{fig.load_path}-\ref{fig.convergence_A-control}. The resulting loading paths practically coincide for both methods and they are in accordance with \cite{NPO08, SCKKZB15, HRS16b}. The computed limit value is equal to 4.057 which is close to the estimate 4.045 known from \cite{CL90}. Both methods need 18 load step and have superlinear convergence in each step. Their convergence is similar up to step 11. However, other comparisons turn out that the indirect method behaves better than the direct one. First, the indirect method has less number of iterations between steps 12 and 18. Secondly, the direct method contained 8 additional load steps without successful convergence while the indirect one convergences in each step. The successful load steps for both methods are depicted  by the circular points in Figures \ref{fig.load_path} and \ref{fig.load_path_A-control}, respectively. We see that the positions of these points are more convenient in Figure \ref{fig.load_path_A-control} than in Figure \ref{fig.load_path} with respect to the curvature of the loading path. Thirdly, we see in Figure \ref{fig.convergence} that the convergence in steps 11 and 16 is superlinear only up to 1e-10. Then, values of the stopping criterion oscillate. This is also observed for a few other steps of the direct method (e.g., steps 14 and 15). For the indirect method, this is not observed at any step. Finally, the computational times of the direct and indirect methods on a current laptop were approximately 9 and 7 minutes, respectively.

\subsection{Associative perfect plastic 3D problem}

Within the 3D slope stability experiment (code  SS-MC-NP-3D), we compare the loading paths for the Q1 and Q2 hexahedral elements with 8 and 20 nodes, respectively. We consider $2\times 2\times 2$ and $3\times 3\times 3$ noded integration quadratures for these element types, respectively. Two hexahedral meshes are prepared for this experiment. For the Q1 elements, the meshes contain 5103 and 37597 nodal points, 34560 and 276480 integration points, respectively. For the Q2 elements, the meshes contain 19581 and 147257 nodal points, 116640 and 933120 integration points, respectively. We use the direct method of the incremental limit analysis which is terminated when the computed settlement exceeds 5 meters.

\begin{figure}[htbp]
\center
  \includegraphics[width=0.5\textwidth]{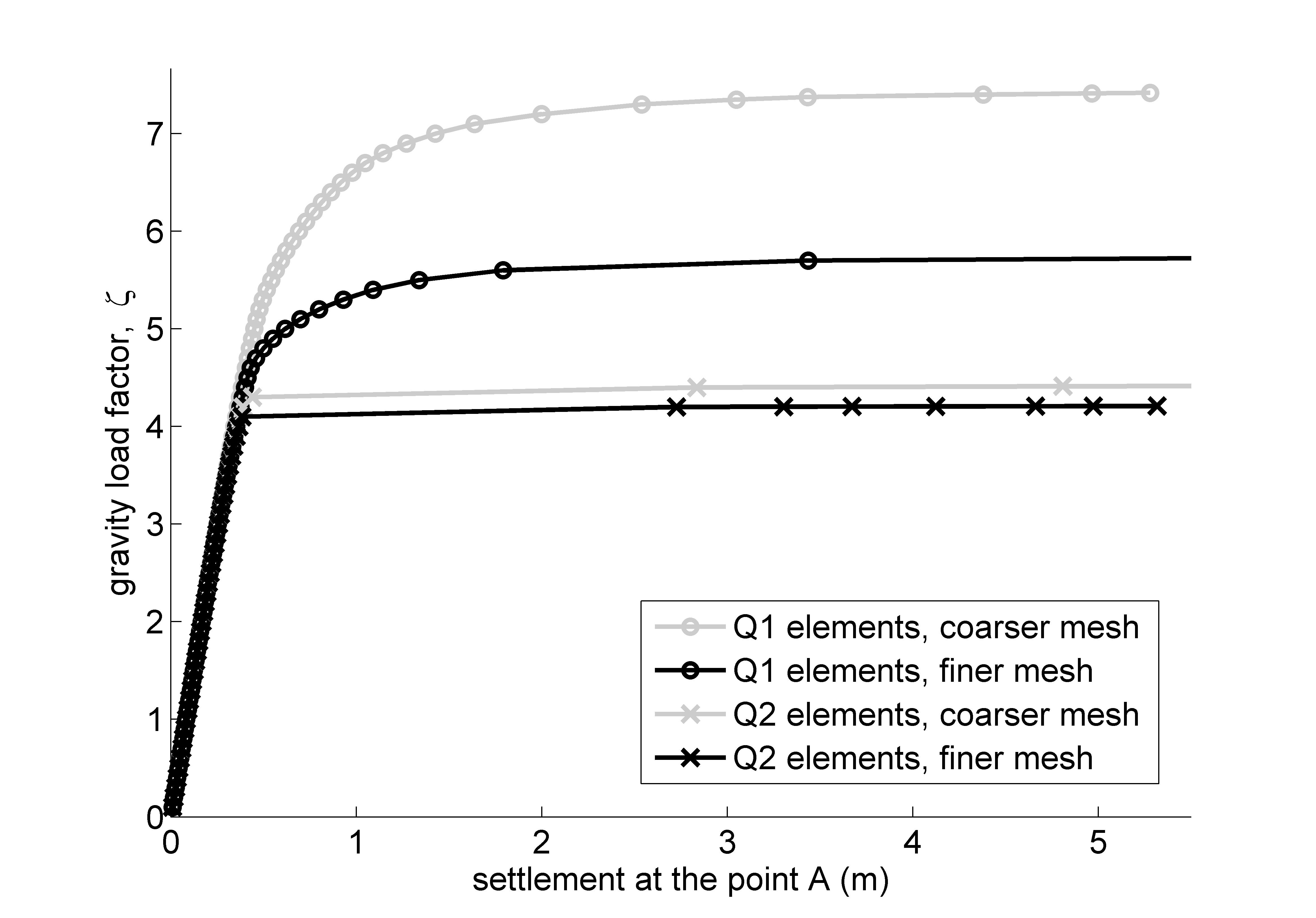}
   \caption{\small{Comparison of the loading paths for $Q1$ and $Q2$ elements.}}
   \label{fig.load_path_3D}
\end{figure}

The corresponding loading paths are depicted in Figure \ref{fig.load_path_3D}. We observe that the estimated limit values of $\zeta$ are close to the expected value of 4.045 for the $Q2$ elements but not for the $Q1$ elements. To estimate $\zeta_{lim}$ using the $Q1$ elements, it would be necessary to use much finer meshes. Figures \ref{fig.displacement} and \ref{fig.multiplier} illustrate failure at the end of the loading process for the Q2 elements and the finer mesh. 

\begin{figure}[htbp]
\begin{minipage}[t]{0.47\textwidth}
  \center
  \includegraphics[width=\textwidth]{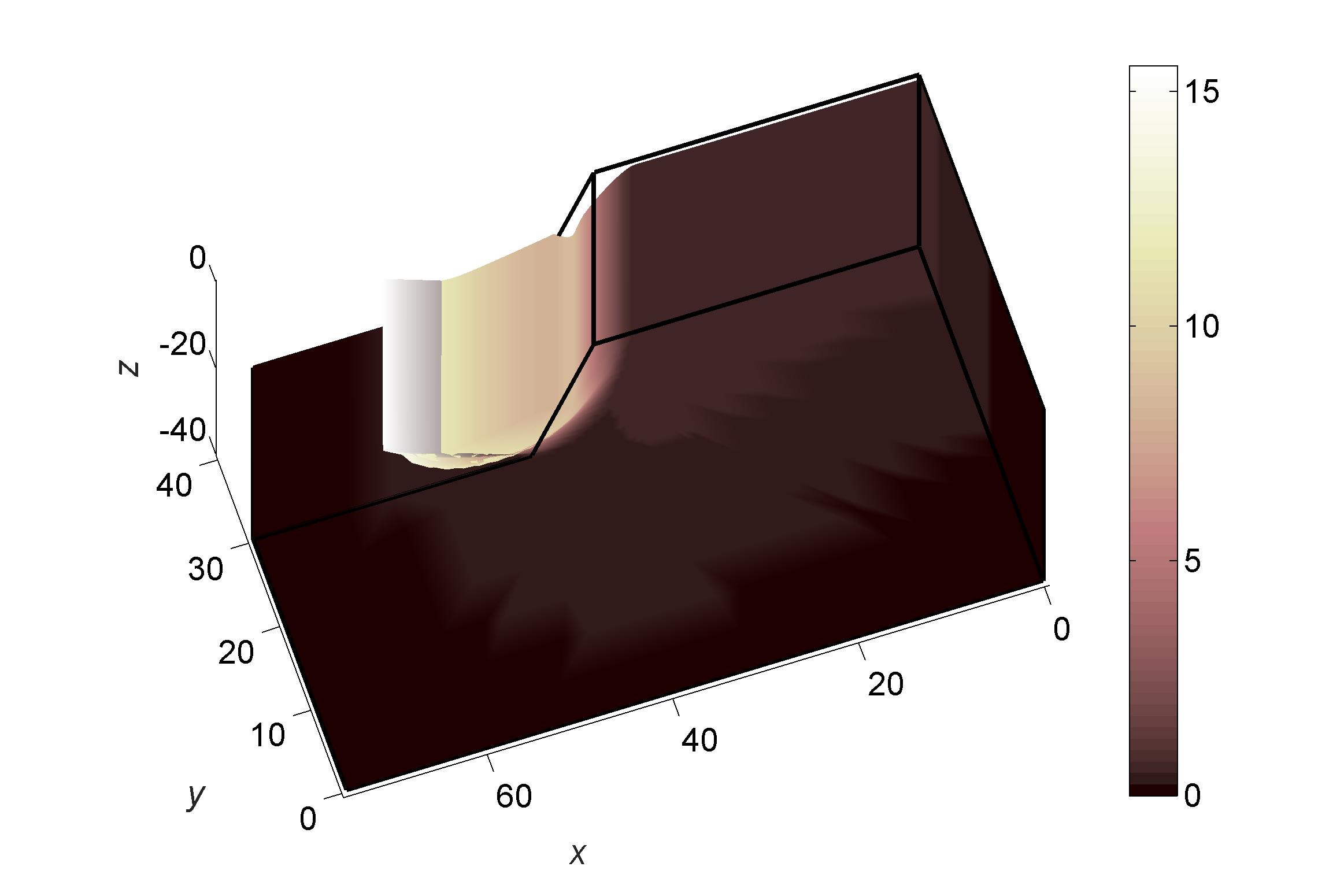}
   \caption{\small{Total displacement and deformed shape at the end of the loading process.}}
   \label{fig.displacement}
\end{minipage}
\hfill
\begin{minipage}[t]{0.47\textwidth}
  \center
   \includegraphics[width=\textwidth]{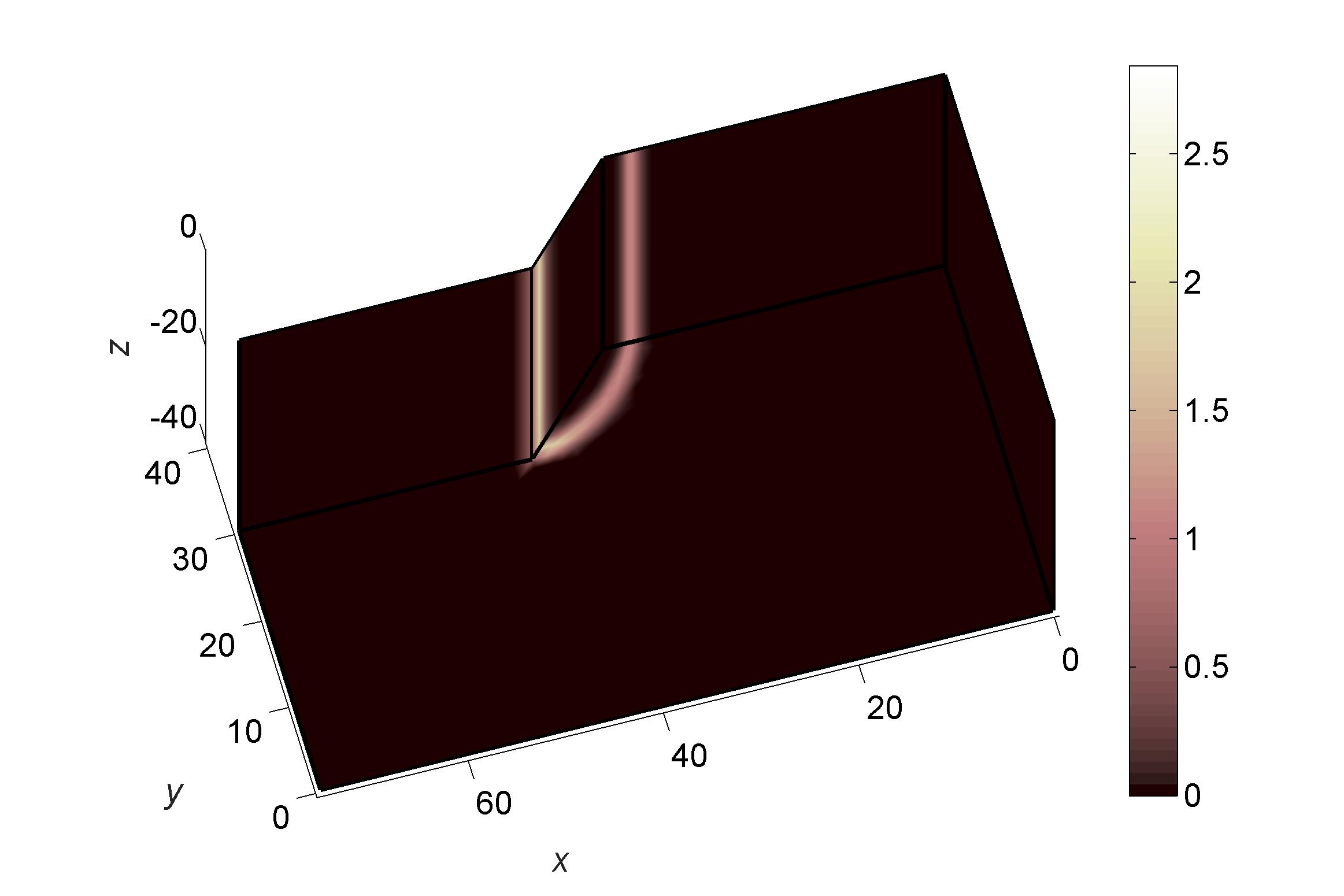}
   \caption{\small{Plastic multipliers at the end of the loading process.}}
   \label{fig.multiplier}
\end{minipage}
\end{figure}

%***********************************************************************************************
\section{Conclusion}
\label{sec_conclusion}

This paper extended the subdifferential-based constitutive solution technique proposed in \cite{SCKKZB15} to elastoplastic models containing the Mohr-Coulomb yield criterion. It enabled deeper analysis of the constitutive problem discretized by the implicit Euler method and consequently led to several improvements within solution schemes. For example, a priori decision criteria characterizing each type of the return-mapping were derived even if the solution could not be found in closed form. The construction of the consistent tangent operator was also simplified. Moreover, the paper brought self-contained derivation of the constitutive operators which is not too often in Mohr-Coulomb plasticity.

The improved constitutive solution schemes were implemented within slope stability problems in 2D and 3D. To this end, the direct and also the indirect methods of the incremental limit analysis were used and combined with the semismooth Newton method. Local superlinear convergence in each step of both methods was observed. Further, it was illustrated that the indirect method led to more stable control of the loading process or that higher order finite elements reduced strong dependence on mesh. 

\section*{Acknowledgements}
The authors would like to thank to Pavel Mar\v s\'alek for generating the quadrilateral meshes with and without midpoints in 2D and 3D.
This work was supported by The Ministry of Education, Youth and Sports (of the Czech Republic)
from the National Programme of Sustainability (NPU II), project
``IT4Innovations excellence in science - LQ1602".

%***********************************************************************************************

\section*{Appendix}
\label{sec_appendix}

\subsection*{A. Simplified constitutive handling for the plane strain problem}

The results of Section \ref{sec_time_discret} and \ref{sec.Stress-strain_relation} are, of course, valid also for the plane strain problem. Nevertheless, in this case, one can simplify the forms of eigenprojections and their derivatives since we work only on the subspace $\mathbb R_{PS}$ of $\mathbb R^{3\times 3}_{sym}$ containing trial tensors in the form
$$\mbf\eta=\left(
\begin{array}{ccc}
\eta_{11} & \eta_{12} & 0\\
\eta_{12} & \eta_{22} & 0\\
0 & 0 & \eta_{33}
\end{array}
\right).
$$
To distinguish the derivatives of functions defined in $\mathbb R_{PS}$, we use the symbol $\tilde{\mathcal D}$ instead of $\mathcal D$.  Define the functions
\begin{eqnarray*}
\tilde\omega_1(\mbf\eta)&:=&\frac{1}{2}\left[\eta_{11}+\eta_{22}+\sqrt{(\eta_{11}-\eta_{22})^2+4\eta_{12}^2}\right],\\
\tilde\omega_2(\mbf\eta)&:=&\frac{1}{2}\left[\eta_{11}+\eta_{22}-\sqrt{(\eta_{11}-\eta_{22})^2+4\eta_{12}^2}\right],\\
\tilde\omega_3(\mbf\eta)&:=&\eta_{33}
\end{eqnarray*}
in $\mathbb R_{PS}$. Then $\tilde\eta_i=\tilde\omega_i(\mbf\eta)$, $i=1,2,3$, are the eigenvalues of $\mbf\eta$. These values are not ordered in general. We only know that $\tilde\eta_1\geq\tilde\eta_2$. Further, define
$$\tilde{\mbf\eta}(\mbf\eta):=\left(
\begin{array}{r r r}
\eta_{11} & \eta_{12} & 0\\
\eta_{12} & \eta_{22} & 0\\
0 & 0 & 0
\end{array}
\right),\quad \tilde{\mbf I}:=\left(
\begin{array}{r r r}
1 & 0 & 0\\
0 & 1 & 0\\
0 & 0 & 0
\end{array}
\right), \quad \tilde{\mbf E}_3(\mbf\eta):=\left(
\begin{array}{r r r}
0 & 0 & 0\\
0 & 0 & 0\\
0 & 0 & 1
\end{array}
\right),$$
$$\tilde{\mbf E}_1(\mbf\eta):=\left\{
\begin{array}{c l}
\frac{\tilde{\mbf\eta}-\tilde\eta_2\tilde{\mbf I}}{\tilde\eta_1-\tilde\eta_2}, & \tilde\eta_1>\tilde\eta_2\\[2mm]
\tilde{\mbf I}, & \tilde\eta_1=\tilde\eta_2
\end{array}
\right.,\quad \tilde{\mbf E}_2(\mbf\eta):=\tilde{\mbf I}-\tilde{\mbf E}_1(\mbf\eta)$$
and
$$\tilde{\mathbb E}_1(\mbf\eta):=\left\{
\begin{array}{c l}
\frac{1}{\tilde\eta_1-\tilde\eta_2}[\tilde{\mathbb I}-\tilde{\mbf E}_1\otimes\tilde{\mbf E}_1-\tilde{\mbf E}_2\otimes\tilde{\mbf E}_2], & \tilde\eta_1>\tilde\eta_2\\[2mm]
\mathbb O, & \tilde\eta_1=\tilde\eta_2
\end{array}
\right.,\quad \tilde{\mathbb E}_2(\mbf\eta):=-\tilde{\mathbb E}_1(\mbf\eta),\quad \tilde{\mathbb E}_3(\mbf\eta):=\mathbb O,$$
where $\mathbb O$ denotes the zeroth fourth order tensor and $[\tilde{\mathbb I}]_{ijkl}=\delta_{ik}\delta_{jl}$, $i,j,k,l=1,2$, otherwise $[\tilde{\mathbb I}]_{ijkl}=0$. Clearly, $\mathcal D\tilde\omega_3(\mbf\eta)=\tilde{\mathcal D}\tilde\omega_3(\mbf\eta)=\tilde{\mbf E}_3(\mbf\eta)$. If $\tilde\eta_1>\tilde\eta_2$ then
$$\tilde{\mbf E}_i(\mbf\eta)=\tilde{\mathcal D}\tilde\omega_i(\mbf\eta),\quad \tilde{\mathbb E}_i(\mbf\eta)=\tilde{\mathcal D}\tilde{\mbf E}_i(\mbf\eta)\qquad\mbox{in } \mathbb R_{PS},\quad i=1,2.$$
It is worth mentioning that these formulas need not hold in $\mathbb R^{3\times 3}_{sym}$ in general. Similar formulas are also introduced in \cite[Appendix A]{NPO08}.

Now, it is necessary to reorder the eigenvalues of $\mbf\eta\in\mathbb R_{PS}$. Denote the ordered eigenvalues as $\eta_1,\eta_2,\eta_3$, i.e., $\eta_1:=\max\{\tilde\eta_1,\tilde\eta_3\}$ and $\eta_3:=\min\{\tilde\eta_2,\tilde\eta_3\}$. Consequently, we reorder the functions $\tilde\omega_i$, $\tilde{\mbf E}_i$, $\tilde{\mathbb E}_i$, $i=1,2,3$, leading to the functions $\omega_i$, $\mbf E_i$, $\mathbb E_i$, $i=1,2,3$. To complete the notation, one can easily set 
$$\mbf E_{12}(\mbf\eta):=\mbf E_1(\mbf\eta)+\mbf E_2(\mbf\eta),\quad \mbf E_{23}(\mbf\eta):=\mbf E_2(\mbf\eta)+\mbf E_3(\mbf\eta)\quad \forall\mbf\eta\in \mathbb R_{PS}.$$

Finally, one can straightforwardly use the functions $\omega_i$, $\mbf E_i$, $\mathbb E_i$, $i=1,2,3$, $\mbf E_{12}$ and $\mbf E_{23}$ within Section \ref{sec.Stress-strain_relation} when the plane strain assumptions are considered.

\subsection*{B. Algebraic representation of second and fourth order tensors}

Within our implementation, we use the standard algebraic representation of stress and strain second order tensors specified below but a little bit different representation of fourth order tensors in comparison to \cite[Appendix D]{NPO08}. We assume that a fourth order tensor $\mathbb C$ represents a linear mapping from $\mathbb R^{3\times 3}_{sym}$ into $\mathbb R^{3\times 3}_{sym}$. Therefore, the components $[\mathbb C]_{ijkl}\equiv C_{ijkl}$  of $\mathbb C$ satisfy
$$\sum_{k,l}C_{ijkl}\eta_{kl}=\sum_{k,l}C_{jikl}\eta_{kl}\quad\forall \mbf\eta\in \mathbb R^{3\times 3}_{sym},\;\;\eta_{kl}=[\mbf\eta]_{kl}.$$
The choice $\eta_{kl}=\delta_{mk}\delta_{nl}+\delta_{nk}\delta_{ml}$ implies that $\mbf\eta\in \mathbb R^{3\times 3}_{sym}$ for any $m,n=1,2,3$ and
\begin{equation}
\tag{B.1}
C_{ijmn}+C_{ijnm}=C_{jimn}+C_{jinm}\quad\forall i,j,m,n=1,2,3.
\label{C_sym}
\end{equation}
Notice that in \cite[Appendix D]{NPO08}, the stronger assumptions on the components are required: $C_{ijmn}=C_{ijnm}=C_{jimn}=C_{jinm}$.

We distinguish two cases: the 3D problem and its plane strain reduction.

\subsubsection*{The 3D problem}

Let $\mbf \tau, \mbf\eta\in \mathbb R^{3\times 3}_{sym}$ denote stress and strain tensors, respectively. Then they are represented by vectors $\mathbf t=( \tau_{11}, \tau_{22}, \tau_{33}, \tau_{12}, \tau_{23}, \tau_{13})^T$ and $\mathbf n=(\eta_{11},\eta_{22},\eta_{33},2\eta_{12},2\eta_{23},2\eta_{13})^T$ where $\tau_{ij}$ and $\eta_{ij}$ are the components of $\mbf \tau$, and $\mbf\eta$, respectively. Clearly, $\mbf \tau:\mbf\eta=\mathbf t\cdot\mathbf n$. A fourth order tensor $\mathbb C$ is represented by matrix $\mbf C\in\mathbb R^{6\times 6}$. Since fourth order tensors are applied on strain tensors within the implementation, we require that
\begin{equation}
\tag{B.2}
\mbf\eta:\mathbb C:\mbf\varepsilon=\mathbf n\cdot\mbf C\mathbf e
\label{C_algebra}
\end{equation}
holds for any strain tensors $\mbf\eta$ and $\mbf\varepsilon$. Here, $\mathbf n$ and $\mathbf e$ denote the algebraic counterparts of $\mbf\eta$ and $\mbf\varepsilon$, respectively. From (\ref{C_sym}) and (\ref{C_algebra}), one can derive that
$$\mbf C=\left(
\begin{array}{c c c c c c}
C_{1111} & C_{1122} & C_{1133} & \frac{1}{2}[C_{1112}+C_{1121}] & \frac{1}{2}[C_{1123}+C_{1132}] & \frac{1}{2}[C_{1113}+C_{1131}]\\[1mm]
C_{2211} & C_{2222} & C_{2233} & \frac{1}{2}[C_{2212}+C_{2221}] & \frac{1}{2}[C_{2223}+C_{2232}] & \frac{1}{2}[C_{2213}+C_{2231}]\\[1mm]
C_{3311} & C_{3322} & C_{3333} & \frac{1}{2}[C_{3312}+C_{3321}] & \frac{1}{2}[C_{3323}+C_{3332}] & \frac{1}{2}[C_{3313}+C_{3331}]\\[1mm]
C_{1211} & C_{1222} & C_{1233} & \frac{1}{2}[C_{1212}+C_{1221}] & \frac{1}{2}[C_{1223}+C_{1232}] & \frac{1}{2}[C_{1213}+C_{1231}]\\[1mm]
C_{2311} & C_{2322} & C_{2333} & \frac{1}{2}[C_{2312}+C_{2321}] & \frac{1}{2}[C_{2323}+C_{2332}] & \frac{1}{2}[C_{2313}+C_{2331}]\\[1mm]
C_{1311} & C_{1322} & C_{1333} & \frac{1}{2}[C_{1312}+C_{1321}] & \frac{1}{2}[C_{1323}+C_{1332}] & \frac{1}{2}[C_{1313}+C_{1331}]
\end{array}
\right).$$
Indeed, the choices $\varepsilon_{kl}=\frac{1}{2}(\delta_{2k}\delta_{3l}+\delta_{2k}\delta_{3l})$, $\eta_{ij}=\frac{1}{2}(\delta_{1i}\delta_{2j}+\delta_{2i}\delta_{1j})$ imply $\mathbf e=(0,0,0,0,1,0)^T$ and $\mathbf n=(0,0,0,1,0,0)^T$. Hence,
$$[\mbf C]_{45}=\mathbf n\cdot \mbf C\mathbf e\stackrel{(\ref{C_algebra})}{=}\eta:\mathbb C:\varepsilon=\frac{1}{4}[C_{1223}+C_{1232}+C_{2123}+C_{2132}]\stackrel{(\ref{C_sym})}{=}\frac{1}{2}[C_{1223}+C_{1232}].$$
Similarly, one can derive the forms of other components of $\mbf C$. Notice that the algebraic representation of $\mathbb C$ is more general than in \cite[Appendix D]{NPO08}.

We introduce three examples useful for the Mohr-Coulomb model:
\begin{enumerate}
\item Let $\mathbb C=\mathbb I$. Then $C_{ijkl}=\delta_{ik}\delta_{jl}$ and $\mbf C=\mbox{diag}(1,1,1,1/2,1/2,1/2)$. Notice that the same matrix is derived in \cite[Appendix D]{NPO08} although the tensor $\mathbb I_{S}$, $[\mathbb I_S]_{ijkl}=\frac{1}{2}(\delta_{ik}\delta_{jl}+\delta_{il}\delta_{jk})$, is used there instead of $\mathbb I$.
\item Let $\mathbb C=\mbf\tau\otimes\mbf\sigma$ where $\mbf\sigma,\mbf\tau$ are arbitrary chosen stress tensors. Denote $\mathbf s$ and $\mathbf t$ as the algebraic counterparts to $\mbf\sigma,\mbf\tau$, respectively. Then $\mbf C=\mathbf s\mathbf t^T$.
\item Let $\mathbb C=\mathcal D(\mbf\eta^2)$. Then $C_{ijkl}=\delta_{ik}\eta_{lj}+\delta_{jl}\eta_{ik}$ and
$$\mbf C=\left(
\begin{array}{c c c c c c}
2\eta_{11} & 0 & 0 & \eta_{12} & 0 & \eta_{13}\\[1mm]
0 & 2\eta_{22} & 0 & \eta_{12} & \eta_{23} & 0\\[1mm]
0 & 0 & 2\eta_{33} & 0 & \eta_{23} & \eta_{13}\\[1mm]
\eta_{12} & \eta_{12} & 0 & \frac{1}{2}[\eta_{11}+\eta_{22}] & \frac{1}{2}\eta_{13} & \frac{1}{2}\eta_{23}\\[1mm]
0 & \eta_{23} & \eta_{23} & \frac{1}{2}\eta_{13} & \frac{1}{2}[\eta_{22}+\eta_{33}] & \frac{1}{2}\eta_{12}\\[1mm]
\eta_{13} & 0 & \eta_{13} & \frac{1}{2}\eta_{23} & \frac{1}{2}\eta_{12} & \frac{1}{2}[\eta_{11}+\eta_{33}]
\end{array}
\right).$$
\end{enumerate}

\subsubsection*{The plane strain problem}

Let $\mbf \tau$ and $\mbf\eta$ denote stress and strain second order tensors, respectively. Then they are represented by the vectors $\mathbf t=( \tau_{11}, \tau_{22}, \tau_{12}, \tau_{33})^T$ and $\mathbf n=(\eta_{11},\eta_{22},2\eta_{12},\eta_{33})^T$ where $\tau_{ij}$ and $\eta_{ij}$ are components of $\mbf \tau$, and $\mbf\eta$, respectively. Clearly, $\mbf \tau:\mbf\eta=\mathbf t\cdot\mathbf n$. Notice that the component $\eta_{33}$ vanishes for the strain tensor but not for the plastic strain tensor.

The fourth order tensor $\mathbb C$ can be represented by matrix $\mbf C\in\mathbb R^{4\times 4}$. Similarly as for the 3D problem, one can derive that
$$\mbf C=\left(
\begin{array}{c c c c}
C_{1111} & C_{1122} & \frac{1}{2}[C_{1112}+C_{1121}] & C_{1133} \\[1mm]
C_{2211} & C_{2222} & \frac{1}{2}[C_{2212}+C_{2221}] & C_{2233} \\[1mm]
C_{1211} & C_{1222} & \frac{1}{2}[C_{1212}+C_{1221}] & C_{1233} \\[1mm]
C_{3311} & C_{3322} & \frac{1}{2}[C_{3312}+C_{3321}] & C_{3333} 
\end{array}
\right).$$
Finally, it is worth mentioning that for assembling the tangent stiffness matrix, it is sufficient to save only the components $(\mbf C)_{ij}$ where $i,j=1,2,3$.

\bibliographystyle{wileyj}
\bibliography{MC_SysCer}

\begin{thebibliography}{10}
\providecommand{\url}[1]{\texttt{#1}}
\providecommand{\urlprefix}{URL }
\expandafter\ifx\csname urlstyle\endcsname\relax
  \providecommand{\doi}[1]{doi:\discretionary{}{}{}#1}\else
  \providecommand{\doi}{doi:\discretionary{}{}{}\begingroup
  \urlstyle{rm}\Url}\fi

\bibitem{SCKKZB15}
Sysala S, Cermak M, Koudelka T, Kruis J, Zeman J, Blaheta R.
  Subdifferential-based implicit return-mapping operators in computational
  plasticity. \emph{ZAMM-Journal of Applied Mathematics and
  Mechanics/Zeitschrift f{\"u}r Angewandte Mathematik und Mechanik}  2016; .

\bibitem{NPO08}
de~Souza~Neto EA, Peri{\'{c}} D, Owen DRJ. \emph{Computational Methods for
  Plasticity}. Wiley-Blackwell, 2008.

\bibitem{CDA15}
Clausen J, Damkilde L, Andersen LV. Robust and efficient handling of yield
  surface discontinuities in elasto-plastic finite element calculations.
  \emph{Engineering Computations}  2015; \textbf{32}(6):1722--1752.

\bibitem{LL15}
Lin C, Li YM. A return mapping algorithm for unified strength theory model.
  \emph{International Journal for Numerical Methods in Engineering}  2015;
  \textbf{104}(8):749--766.

\bibitem{LR96}
Larsson R, Runesson K. Implicit integration and consistent linearization for
  yield criteria of the mohr--coulomb type. \emph{Mechanics of
  Cohesive-frictional Materials}  1996; \textbf{1}(4):367--383.

\bibitem{BSS03}
Borja RI, Sama KM, Sanz PF. On the numerical integration of three-invariant
  elastoplastic constitutive models. \emph{Computer Methods in Applied
  Mechanics and Engineering}  2003; \textbf{192}(9-10):1227--1258.

\bibitem{K13}
Karaoulanis FE. Implicit numerical integration of nonsmooth multisurface yield
  criteria in the principal stress space. \emph{Arch Computat Methods Eng}
  2013; \textbf{20}(3):263--308.

\bibitem{K53}
Koiter WT. Stress-strain relations, uniqueness and variational theorems for
  elastic-plastic materials with a singular yield surface. \emph{Quarterly of
  Applied Mathematics}  1953; \textbf{11}(3):350--354.

\bibitem{dB87}
de~Borst R. Integration of plasticity equations for singular yield functions.
  \emph{Computers {\&} Structures}  1987; \textbf{26}(5):823--829.

\bibitem{HR99}
Han W, Reddy BD. \emph{Plasticity: mathematical theory and numerical analysis}.
  Springer-Verlag, 1999.

\bibitem{B12}
Berga A. Mathematical and numerical modeling of the non-associated plasticity
  of soils{\textemdash}part 1: The boundary value problem. \emph{International
  Journal of Non-Linear Mechanics}  2012; \textbf{47}(1):26--35.

\bibitem{SH98}
Simo JC, Hughes TJ. \emph{Computational inelasticity}. Springer Science \&
  Business Media, 2006.

\bibitem{SHVS14}
Starman B, Halilovi{\v{c}} M, Vrh M, {\v{S}}tok B. Consistent tangent operator
  for cutting-plane algorithm of elasto-plasticity. \emph{Computer Methods in
  Applied Mechanics and Engineering}  2014; \textbf{272}:214--232.

\bibitem{CDA06}
Clausen J, Damkilde L, Andersen L. Efficient return algorithms for associated
  plasticity with multiple yield planes. \emph{International Journal for
  Numerical Methods in Engineering}  2006; \textbf{66}(6):1036--1059.

\bibitem{ALSH11}
Abbo A, Lyamin A, Sloan S, Hambleton J. A {C}2 continuous approximation to the
  {M}ohr{\textendash}{C}oulomb yield surface. \emph{International Journal of
  Solids and Structures}  2011; \textbf{48}(21):3001--3010.

\bibitem{B13}
Borja RI. \emph{Plasticity}. Springer, 2013.

\bibitem{CKSV14}
{\v{C}}erm{\'{a}}k M, Kozubek T, Sysala S, Valdman J. A {TFETI} domain
  decomposition solver for elastoplastic problems. \emph{Applied Mathematics
  and Computation}  2014; \textbf{231}:634--653.

\bibitem{GrVa09}
Gruber PG, Valdman J. Solution of one-time-step problems in elastoplasticity by
  a slant {N}ewton method. \emph{{SIAM} J. Sci. Comput.}  2009;
  \textbf{31}(2):1558--1580.

\bibitem{SaWi11}
Sauter M, Wieners C. On the superlinear convergence in computational
  elasto-plasticity. \emph{Computer Methods in Applied Mechanics and
  Engineering}  2011; \textbf{200}(49-52):3646--3658.

\bibitem{Sy09}
Sysala S. Application of a modified semismooth {N}ewton method to some
  elasto-plastic problems. \emph{Mathematics and Computers in Simulation}
  2012; \textbf{82}(10):2004--2021.

\bibitem{Sy14}
Sysala S. Properties and simplifications of constitutive time-discretized
  elastoplastic operators. \emph{{ZAMM} - Journal of Applied Mathematics and
  Mechanics / Zeitschrift f\"{u}r Angewandte Mathematik und Mechanik}  2014;
  \textbf{94}(3):233--255.

\bibitem{CaHo86}
Carlson DE, Hoger A. The derivative of a tensor-valued function of a tensor.
  \emph{The Quarterly of Applied Mathematics}  1986; \textbf{44}:409--423.

\bibitem{C97}
de~Borst R, Crisfield MA, Remmers JJC, Verhoosel CV. \emph{Non-Linear Finite
  Element Analysis of Solids and Structures}. Wiley-Blackwell, 2012.

\bibitem{CL90}
Chen WF, Liu X. \emph{Limit analysis in soil mechanics}. Elsevier, 2012.

\bibitem{SHHC15}
Sysala S, Haslinger J, Hlav{\'a}{\v{c}}ek I, Cermak M. Discretization and
  numerical realization of contact problems for elastic-perfectly plastic
  bodies. {PART} {I} - discretization, limit analysis. \emph{ZAMM-Journal of
  Applied Mathematics and Mechanics/Zeitschrift f{\"u}r Angewandte Mathematik
  und Mechanik}  2015; \textbf{95}(4):333--353.

\bibitem{CHKS15}
Cermak M, Haslinger J, Kozubek T, Sysala S. Discretization and numerical
  realization of contact problems for elastic-perfectly plastic bodies. {PART}
  {II} - numerical realization, limit analysis. \emph{{ZAMM} - Journal of
  Applied Mathematics and Mechanics / Zeitschrift f\"{u}r Angewandte Mathematik
  und Mechanik}  2015; \textbf{95}(12):1348--1371.

\bibitem{HRS16}
Haslinger J, Repin S, Sysala S. A reliable incremental method of computing the
  limit load in deformation plasticity based on compliance: Continuous and
  discrete setting. \emph{Journal of Computational and Applied Mathematics}
  2016; \textbf{303}:156--170.

\bibitem{HRS16b}
Haslinger J, Repin S, Sysala S. Guaranteed and computable bounds of the limit
  load for variational problems with linear growth energy functionals.
  \emph{Applications of Mathematics}  2016; \textbf{61}(5):527--564.

\bibitem{Ru06}
Ruszczy{\'n}ski AP. \emph{Nonlinear optimization}. Princeton university press,
  2006.

\bibitem{Mcode}
Sysala S, Cermak M. Experimental matlab code for the slope stability benchmark
  -- {SS-MC-NP-3D}, {SS-MC-NH}, {SS-MC-NP-A}control 2016.
  \urlprefix\url{www.ugn.cas.cz/?p=publish/output.php}, (or www.ugn.cas.cz -
  Publications - Other outputs - SS-MC-NP-3D, SS-MC-NH, SS-MC-NP-Acontrol).

\end{thebibliography}

\end{document}